\documentclass{amsart}
\usepackage{mathrsfs}
\usepackage{graphicx}
\usepackage{amsmath}
\usepackage{amscd}
\usepackage{amssymb}

\input xy
\xyoption{all}

\newcommand{\pd}{\partial}
\newcommand{\bC}{{\mathbb C}}

 \newcommand{\bs}{{\mathbf s}}
 
\newcommand{\bR}{{\mathbb R}} 
\newcommand{\bZ}{{\mathbb Z}}

 \newcommand{\cD}{{\mathcal D}}
\newcommand{\cE}{{\mathcal E}}

\newcommand{\cM}{{\mathcal M}}

\newcommand{\half}{\frac{1}{2}}

\newcommand{\Mbar}{\overline{\cM}}

 \DeclareMathOperator{\Res}{Res}

\newtheorem{theorem/definition}{Theorem/Definition}[section]
\newtheorem{thm}{Theorem}
\newtheorem{prop}{Proposition}[section]
\newtheorem{lm}{Lemma}[section]

\newtheorem{corollary}{Corollary}[section]

\newtheorem{rmk}{Remark}

\theoremstyle{remark}
\newtheorem{remark}{Remark}[section]

\theoremstyle{definition}
 \newtheorem{example}{Example}[section]

\newcommand{\be}{\begin{equation}}
\newcommand{\ee}{\end{equation}}
\newcommand{\bea}{\begin{eqnarray}}
\newcommand{\ben}{\begin{eqnarray*}}
\newcommand{\een}{\end{eqnarray*}}
\newcommand{\eea}{\end{eqnarray}}
\newcommand{\bet}{\begin{equation}
\begin{split}}
\newcommand{\eet}{\end{split}
\end{equation}}

\usepackage{color}

\definecolor{yellow}{rgb}{1,1,0}
\definecolor{orange}{rgb}{1,.7,0}
\definecolor{red}{rgb}{1,0,0} \definecolor{green}{rgb}{0,1,1}
\definecolor{white}{rgb}{1,1,1}

\definecolor{A}{rgb}{.75,1,.75}

\newcommand{\corr}[1]{\langle {#1} \rangle}

\newcommand{\Corr}[1]{\biggl\langle {#1} \biggr\rangle}

\begin{document}

\title
{Emergent Geometry of Matrix Models with Even Couplings}

\author{Jian Zhou}
\address{Department of Mathematical Sciences\\
Tsinghua University\\Beijing, 100084, China}
\email{jianzhou@mail.tsinghua.edu.cn}

\begin{abstract}
We show that to the modified GUE partition function with even coupling introduced by 
Dubrovin, Liu, Yang and Zhang, one can associate $n$-point correlation functions
in arbitrary genera which satisfy Eynard-Orantin topological recursions.
Furthermore,
these $n$-point functions are related to intersection numbers on the Deligne-Mumford moduli spaces.    
\end{abstract}
\maketitle

\section{Introduction}

In the early 1990s there appeared two different approaches to 2D topological gravity.
The first one is via the double scaling limits of large N Hermitian one-matrix models.
In this approach a connection to KdV hierarchy and Virasoro constraints \cite{DVV} is established. 
Witten \cite{Witten} introduced another approach via the intersection theory on
the Deligne-Mumford moduli spaces $\Mbar_{g,n}$ of stable algebraic curves \cite{Witten}.
He conjectured a famous connection between intersection numbers on $\Mbar_{g,n}$
and KdV hierarchy and Virasoro constraints
based on the conjectural equivalence between these two approaches.
Of course,
at the same time,
a connection between matrix models and intersection numbers of moduli spaces $\Mbar_{g,n}$ of Riemann surfaces
was suggested.
In the mathematical literature,
a connection between matrix models and the orbifold Euler characteristics of $\cM_{g,n}$ 
appeared earlier in the works of Harer-Zagier \cite{Har-Zag} and Penner \cite{Penner}. 

Witten's conjecture was proved by Kontsevich \cite{Kontsevich} by introducing 
a different kind of matrix models, now called the Kontsevich model.
His result establishes a connection between Kontsevich model
with intersection numbers on $\Mbar_{g,n}$.
We are interested in establishing a connection between the original 
Hermitian one-matrix models 
and intersection theory on $\Mbar_{g,n}$, without taking any double scaling limit.

A  result of this type has already appeared recently in the work of Dubrovin, Liu, Yang, and Zhang \cite{DLYZ},
and so we will first examine their results in this paper to test our idea for the general case.
They considered some modified partition function of Hermitian one-matrix model
with only even couplings and identified with the generating series of some special Hodge integrals.
Our goal is to consider the Hermitian one-matrix model with all possible couplings.
We will achieve this goal in a subsequent paper which is modelled on this paper.

The method of \cite{DLYZ} is as follows.
Starting with the Virasoro constraints for GUE partition functions with all possible couplings,
they derived the Virasoro constraints for modified partition function of Hermitian one-matrix model
with only even couplings.
On the other hand, 
they made a cleaver modifications of some Virasoro constraints for Hodge integrals derived by the author in an unpublished
paper to get the same Virasoro constraints for some special Hodge integrals.  
The success of this method depending on knowing the intersection numbers (in this case, 
some Hodge integrals) to be identified with.
In this paper we will use an approach without this knowledge a priori.  
This will be particularly useful in our subsequent paper when 
we study Hermitian one-matrix model with all couplings,
without knowing in advance what type of intersections numbers will be involved. 

We will start with the Virasoro constraints derived in \cite{DLYZ} for modified GUE partition function with even couplings.
Our approach is based on the topological recursions developed by Eynard-Orantin \cite{EO},
its connection with intersection numbers discovered by Eynard \cite{Eynard}
inspired by results in the theory of matrix models,
and the idea of emergent geometry from Virasoro constraints developed by the author \cite{Zhou2, Zhou3}. 
Our key result (Theorem \ref{Thm:Main1}) is that not only the spectral curve,
but also the process of topological recursions, emerge from the Virasoro constraints.
This is exactly in the same spirit as our treatment of  the case of Witten-Konsevich tau-function 
in \cite{Zhou1}.
The spectral curve and the Bergman kernel emerge as a result of studying the genus zero one-point function
and two-point functions respectively,
and the Eynard-Orantin topological recursions emerge as one reformulate the Virasoro 
constraints in  terms of residues on the spectral curve. 
As a result,
we can use Eynard's result to relate the corresponding $n$-point functions to intersection numbers
(cf. Theorem \ref{Thm:Main2}).
We leave the problem of rederiving the result of \cite{DLYZ}
from our approach to future investigations.

As for the case of Hermitian one-matrix model with all couplings,
the emergence of the spectral curve 
has been addressed in \cite{Zhou3}.
In \cite{Zhou4} we will discuss the corresponding emergence of the topological recursions. 
In that case since the spectral curve has two branchpoints,
we will use the generalization of \cite{Eynard} made by Eynard himself in \cite{Eynard2}
to relate to intersection numbers.

\section{Eynard-Orantin Topological Recursions and Intersections on $\Mbar_{g,n}$}
\label{sec:Eynard}

In this Section we recall the $n$-point multilinear differentials $\omega_{g,n}$
obtained from Eynard-Orantin topological recursions
and their  relationship  with the intersections numbers on $\Mbar_{g,n}$.

\subsection{Eynard-Orantin topological recursions}

Recall a  {\em spectral curve} is a parameterized curve
\begin{align}
x& =x(z), & y & =y(z)
\end{align}
together with a Bergman kernel $B(z_1,z_2)$ with the following property
\be
B(z,z') \sim \biggl(\frac{1}{(z-z')^2}  + O(1)\biggr) dz dz'.
\ee
For the purpose of this work,
assume that $x$ has only one  nondegenerate critical point $a$:
\begin{align}
x'(a) & = 0, & x''(a) \neq 0.
\end{align}
Near $a$ there is an involution $\sigma$ such that $x(\sigma(z)) = x(z)$.
A sequence $\omega_{g,n}(z_1, \dots , z_n)$ of multidifferential with $n \geq 1$
and $g \geq 0$ are defined by Eynard-Orantin \cite{EO} as follows:
\ben
&& \omega_{0,1}(z) = y(z) dx(z), \\
&& \omega_{0,2}(z, z_0) = B(z, z_0), \
\een
and for $2g - 2 + n > 0$,
\be \label{eqn:EO}
\begin{split}
& \omega_{g,n+1}(z_0,z_1, \dots, z_n) =
\Res_{z\to a}
K(z_0, z) \biggl[ \omega_{g-1, n+2}(z, \sigma(z), z_{[n]}) \\
+ & \sum^g_{h=0} \sum_{I \subset [n]}
\omega_{h,|I|+1}(z, z_I) \omega_{g-h, n-|I|+1}(\sigma(z), z_{[n]-I})
\biggr],
\end{split}
\ee
the recursion kernel $K$ near $a$  is:
\be
K(z_0, z) =
 \frac{\int_{z'=\sigma(z)}^{z} B(z_0, z')}{2(y(z)-y(\sigma(z))) dx(z)}.
\ee

\subsection{Expansions near the branch point}

Near the branch point $z=a$,
one can introduce a new local coordinate $\zeta$:
\be
\zeta(z) = \sqrt{x(z) - x(a)},
\ee
In other words,
\be
x = x(a) +  \zeta^2.
\ee
This local coordinate is called the {\em local Airy coordinate}.

With the introduction of the local Airy coordinates,
one can express everything involved in the Eynard-Orantin recursion in terms of it.
For $z$ near $a$,
let $\sigma(z)$ denote the unique point near $a$ such that
\be
\zeta(\sigma(z)) = - \zeta(z).
\ee
The function $y$ can be expanded in Taylor series:
\be
y(z) \sim \sum_{k=0} t_{k+2} \zeta^k.
\ee
Similarly, the Bergman kernel can be expanded:
\be \label{eqn:B-Expansion}
B(z,z')
= \biggl[ \frac{1}{(\zeta-\zeta')^2}
+ \sum_{k, l} B_{k,l}   \zeta(z_1)^k \zeta(z_2)^l \biggr]
d\zeta(z_1) d\zeta(z_2).
\ee

The differential $d\zeta_k(z)$ is defined by:
\be
d\zeta_k(z) = - \frac{(2k - 1)!!}{2^k} \Res_{z' \to a} B(z, z') \frac{1}{\zeta(z')^{2k+1}}.
\ee
From the expansion \eqref{eqn:B-Expansion} one gets:
\be \label{eqn:d-zeta-k}
d\zeta_k(z) = - \frac{(2k + 1)!! d\zeta(z)}{2^k \zeta(z)^{2k+2}}
- \frac{(2k - 1)!!}{2^k} \sum_l B_{2k,l} \zeta(z)^l d\zeta(z),
\ee
therefore,
$d\zeta_k(z)$ is the differential of the function $\zeta_k(z)$ defined by:
\be \label{eqn:zeta-k}
\zeta_k(z) = \frac{(2k - 1)!!}{2^k} \biggl(
\frac{1}{\zeta(z)^{2k+1}}
- \sum_l B_{2k, l}  \frac{\zeta(z)^{l+1}}{l + 1}  \biggr).
\ee

\subsection{Laplace transform and intersection numbers}

To make connection with intersection numbers on $\Mbar_{g,n}$,
one needs to perform the Laplace transforms of $\omega_{g,n}$.

The Laplace transform of the $1$-form $ \omega_{0,1} =ydx$ gives the times $\tilde{t}_k$:
\be
e^{- \sum_k \tilde{t}_k u^{-k}} = \frac{2u^{3/2}e^{ux(a)}}{\sqrt{\pi}}
\int_\gamma   e^{-ux} ydx,
\ee
where $\gamma$ is a steepest descent path from the branchpoint to $x = + \infty$, i.e.
$x(\gamma) - x(a) = \bR_+$.
More precisely
\ben
e^{- \sum_k \tilde{t}_k u^{-k}}
& = & \frac{2u^{3/2}e^{ux(a)}}{\sqrt{\pi}}
\int_{-\infty}^\infty    e^{-u(x(a) + \zeta^2) } \sum_{k \geq 0} t_{k+2} \zeta^kd \zeta^2 \\
& = & \sum_k
(2k + 1)!! t_{2k+3} u^{-k}.
\een

The double Laplace transform of the Bergman kernel  gives the kernel
\be
\hat{B}(u,v) = \sum_{k,l} \hat{B}_{k,l} u^k v^l
\ee
by
\be
\begin{split}
& \sum_{k,l} \hat{B}_{k,l} u^k v^l \\
= & \frac{(uu')^{1/2} e^{(u+v) x(a)}}{2\pi}
\int_{z\in \gamma} \int_{z' \in\gamma } e^{- ux(z)} e^{-u'x(z')}
\biggl(B(z, z') -  B_0(z, z')\biggr),
\end{split}
\ee
where the integral is regularized by substracting the  singular part with double pole
\ben
B^0(z_1, z_2) = \frac{dx(z_1)dx(z_2)}{4\sqrt{x(z_1) - x(a)} \sqrt{x(z_2) - x(a)}}
\frac{1}{(\sqrt{x(z_1) - x(a)} -  \sqrt{x(z_2) - x(a)})^2}.
\een
More concretely,
\be \label{eqn:hat-Bkl}
\hat{B}_{k,l} = (2k - 1)!! (2l - 1)!! 2^{-k-l-1} B_{2k,2l}.
\ee

The main result of Eynard \cite{Eynard} is that
the Laplace transforms of $\omega_{g,n}$ when $2g-2+n > 0$ are given by
intersection numbers on $\Mbar_{g,n}$:
\be
\begin{split}
& \prod_{i=1}^n  \sqrt{\frac{\mu_i}{\pi}} e^{\mu_i x(a)}
\int_{z_1\in \gamma} \cdots \int_{z_2 \in \gamma}
\prod_{i=1}^n e^{-\mu_ix(z_i)} \omega_{g,n}(z_1, \dots, z_n) \\
= & 2^{3g-3+2n} \Corr{\prod_{i=1}^n \hat{B}(\mu_i, 1/\psi_i) e^{1/2 \sum_{\delta} \sum_{k,l} \hat{B}_{k,l}
\iota_{\delta*} \psi^k \psi'^l} e^{\sum_k \tilde{t}_k\kappa_k}}_{g,n}
\end{split}
\ee
Equivalently,
\be \label{eqn:TR-Int}
\begin{split}
& \omega_{g,n}(z_1, \dots, z_n) \\
= & 2^{3g-3+n} \sum_{d_1+\cdots +d_n \leq 3g-3+n}
\prod_i
d\zeta_{d_i}(z_i) \Corr{e^{\frac{1}{2} \sum_\delta \iota_{\delta*} \hat{B}(\psi, \psi')}
e^{\sum_k \tilde{t}_k \kappa_k} \prod_i \psi_i^{d_i}}_{g,n}
\end{split}
\ee
When $n=0$:
\be
F_g = 2^{3g-3}  \Corr{e^{\frac{1}{2} \sum_\delta \iota_{\delta*} \hat{B}(\psi, \psi')}
e^{\sum_k \tilde{t}_k \kappa_k}}_{g,0}
\ee

The following are some examples from \cite{Eynard}:
\be \label{eqn:omega03}
\omega_{0,3} (z_1, z_2, z_3) = e^{\tilde{t}_0} d\zeta_0(z_1) d\zeta_0(z_2) d\zeta_0(z_3)
= \frac{1}{2t_3}
d\zeta_0(z_1) d\zeta_0(z_2) d\zeta_0(z_3).
\ee
\be \label{eqn:omega11}
\omega_{1,1}(z) = \frac{1}{24 t_3} d\zeta_1(z)
+ \biggl(\frac{B_{0,0}}{4t_3}- \frac{t_5}{16t_3}
\biggr) d\zeta_0(z),
\ee
\be \label{eqn:omega04}
\begin{split}
\omega_{0,4}(z_1, z_2, z_3, z_4)
=& \frac{1}{2t^2_3}
(d\zeta_1(z_1)d\zeta_0(z_2)d\zeta_0(z_3)d\zeta_0(z_4)
+ perm.) \\
+ & \frac{3(B_{0,0} - t_5/t_3)}{4t_3^2}
d\zeta_0(z_1)d\zeta_0(z_2)d\zeta_0(z_3)d\zeta_0(z_4).
\end{split}
\ee

The following  are some well-known special cases from \cite{Eynard}.

\begin{example}
When the spectral curve is the Airy curve given by the parameterization
\begin{align*}
x& =\half z^2,  & y=z,
\end{align*}
with the Bergman kernel
$$B(x_1,x_2) = \frac{dx_1dx_2}{(x_1-x_2)^2},$$
one has
\be
\omega_{g,n}(z_1, \dots, z_n)
= (-1)^n \sum_{d_1+\cdots +d_n =3g-3+n} \prod_i
\frac{(2d_i + 1)!! d\zeta_i}{\zeta_i^{2d_i+2}}
\cdot \Corr{\prod_i \psi_i^{d_i}}_{g,n}.
\ee
The Eynard-Orantin topological recursions in this case
are equivalent to the DVV Virasoro constraints for Witten-Kontsevich tau-function \cite{Zhou1}.
\end{example}

\begin{example}

When the spectral curve is the Airy curve given by the parameterization
\begin{align*}
x& =\half z^2,  & y=z+ \sum_{k=2}^\infty u_k z^k,
\end{align*}
with the Bergman kernel
$$B(x_1,x_2) = \frac{dx_1dx_2}{(x_1-x_2)^2},$$
one has
\be
\begin{split}
 & \omega_{g,n}(z_1, \dots, z_n) \\
= & (-1)^n \sum_{d_1+\cdots +d_n =3g-3+n} \prod_i
\frac{(2d_i + 1)!! d\zeta_i}{\zeta_i^{2d_i+2}}
\cdot \Corr{\exp \sum_{n=1}^\infty s_n \kappa_n \cdot \prod_i \psi_i^{d_i}}_{g,n},
\end{split}
\ee
where the parameters $\{s_n\}_{n \geq 1}$ are related to $\{u_k\}_{k \geq 2}$ by
\be
\exp (-\sum_{n=1}^\infty s_n z^n) = \sum_{m \geq 1} (2m+1)!! u_{2m+1} z^{2m}.
\ee
The Eynard-Orantin topological recursions in this case 
are related to higher Weil-Peterssson volumes in \cite{Eynard0}.
They are shown to be equivalent to the  Virasoro constraints for higher Weil-Petersson volumes 
proved by Mulase-Safnuk \cite{MS} and Liu-Xu \cite{LX}
in \cite{Zhou1}.
\end{example}

\section{Computing Correlation Functions of Modified Hermitian One-Matrix Model with Even Coupling Constants by Virasoro Constraints}

\label{sec:Correlation}

In this Section we define the $n$-point functions of the modified partition function
of Hermitian one-matrix model and derive an algorithm to compute them
by Virasoro constraints.

\subsection{Modified GUE partition function and modified Virasoro constraints}

Let $Z_{even}$ denote the GUE partition function with even couplings.
Define $\tilde{Z}$ by
\be
\log Z_{even} = \big( \Lambda^{1/2} + \Lambda^{-1/2} \big)  \log \tilde{Z}.
\ee 
Then $\tilde{Z}$ satisfies the followings system of equations
which are called the Virasoro constraints:
\bea
&& \half \frac{\pd F_g}{\pd s_2} =  \sum_{k\geq 1}
k s_{2k} \frac{\pd F_g}{\pd s_{2k}} + \frac{t^2\delta_{g,0}}{4}- \frac{\delta_{g,1}}{16},  \label{eqn:L0} \\
&& \half \frac{\pd F_g}{\pd s_{2n+2}} =
\sum^{n-1}_{k=1} \frac{\pd F_g}{\pd s_{2k}}\frac{\pd F_g}{\pd s_{2n-2k}}
+ \sum^{n-1}_{k=1} \frac{\pd^2 F_g}{\pd s_{2k} \pd s_{2n-2k}}
+ t \frac{\pd F_g}{\pd s_{2n}} +  \sum_{k\geq 1}
ks_{2k} \frac{\pd F_g}{\pd s_{2k+2n}},
\eea
$n \geq 1$,
where 
\be
\log \tilde{Z} = \sum_{g\geq 0} \epsilon^{2g-2} F_g.
\ee
In the above $t=N\epsilon$, and $\Lambda = e^{\epsilon \pd_t}$.
See \cite{DLYZ} for notations and the proof of the above results.

\subsection{Virasoro constraints in terms of  correlators}

Define the $n$-point correlators by
\be
\corr{p_{2a_1} \cdots p_{2a_n} }^c_g
: =\frac{\pd^n}{\pd s_{2a_1} \cdots \pd s_{2a_n}}F_g\biggl|_{s_{2k}=0, k\geq 1}.
\ee

Then the equation  \eqref{eqn:L0} can be written in terms of the correlators as follows:
\be
\half \corr{p_2 \cdot p_{2a_1} \cdots p_{2a_n} }^c_g =
\sum_{j=1}^n a_j \cdot
\corr{p_{2a_1} \cdots p_{2a_n}}_{g},
\ee
together with initial value:
\begin{align}
\corr{p_2}_0 & = \frac{t^2}{2}, &
\corr{p_2}_1 & = - \frac{1}{8}.
\end{align}
And for $m \geq 1$,
\be \label{eqn:P2m}
\begin{split}
& \half\corr{p_{2m+2} \cdot p_{2a_1} \cdots p_{2a_n}}^c_{g}
=\sum_{j=1}^n a_j \cdot \corr{p_{2a_1} \cdots
p_{2a_j+2m} \cdots p_{2a_n} }_{g}^c  \\
& \qquad + t \corr{p_{2m} \cdot p_{2a_1} \cdots p_{2a_n}}^c_{g}
+ \sum_{k=1}^{m-1} \corr{p_{2k}p_{2m-2k} \cdot
p_{2a_1} \cdots p_{2a_n} }^c_{g-1}  \\
& \qquad + \sum_{k=1}^{m-1} \sum_{\substack{g_1+g_2=g\\I_1 \coprod I_2 = [n]}}
\corr{p_{2k} \cdot \prod_{i\in I_1} p_{2a_i}}^c_{g_1}
\cdot \corr{p_{2m-2k} \cdot \prod_{i\in I_2} p_{2a_i}}^c_{g_2},
\end{split}
\ee
where $[n]=\{1, \dots, n\}$.

For example,
\ben
&& \corr{p_2^2}_0^c = 2 \corr{p_2}_0^c = t^2, \\
&& \corr{p_2^3}_0^c = 2 \cdot 2 \corr{p_2 p_2}_0^c = 4 t^2, \\
&& \corr{p_4}_0^c = 2t \corr{p_2}_0^c = t^3.
\een

Now we define the genus $g$, $n$-point functions by:
\be
G_{g,n}(z_1, \dots, z_n) : = \sum_{a_1, \dots,a_n \geq 0} 
\corr{p_{2a_1} \cdots p_{2a_n}}_{g,n}^c \frac{1}{z_1^{a_1+1}} \cdots \frac{1}{z_n^{a_n+1}}.
\ee
Our goal is to compute these functions and interpret the result gemetrically.

\subsection{Computation of genus zero one-point function by Virasoro constraints}

From
\ben
&& \corr{p_2}_0^c = t^2/2, \\
&& \corr{p_{2m}}_0^c = 2\sum_{k=1}^{m-2}
\corr{p_{2k}}_0^c \cdot \corr{p_{2m-2-2k}}_0^c
+ 2 t \cdot \corr{p_{2m-2}}_{0}^c, \quad m \geq 2.
\een
we  get:
\ben
G_{0,1}(x) & = &   \frac{t^2}{2x^2}
+ \sum_{m \geq 2} \frac{1}{x^{m+1}} \biggl(2\sum_{k=1}^{m-2}
\corr{p_{2k}}_0^c \cdot \corr{p_{2m-2-2k}}_0^c
+ 2 t \cdot \corr{p_{2m-2}}_{0}^c \biggr) \\
& = & \frac{t^2}{2x^2} + \frac{2t}{x} G_{0,1}(x)
+2 G_{0,1}(x)^2.
\een
After solving for $G_{0,1}(x)$ we then get:
\be \label{eqn:G01}
G_{0,1}(x) = \frac{1}{4}
\biggl(1-\frac{2t}{x}
- \sqrt{\biggl(1-\frac{2t}{x}\biggr)^2-\frac{4t^2}{x^2}}\biggr)
= \frac{1}{4}
\biggl(1-\frac{2t}{x}
- \sqrt{1-\frac{4t}{x}}\biggr).
\ee
It is then easy to see that:
\be
G_{0,1}(x) = \frac{1}{4}
\sum_{n \geq 2} \frac{(2n-3)!!}{n!}
\cdot (2t)^{n}  x^{-n}
= \frac{1}{2} \sum_{n \geq 2} \frac{(2n-2)!}{(n-1)!n!} \frac{t^n}{x^n}.
\ee
The following are the first few terms of $G_{0,1}(x)$:
\ben
G_{0,1}(x)
& = & \frac{1}{2}\biggl(1 \cdot \frac{t^2}{x^2} + \frac{2t^3}{x^3}
+ \frac{5t^4}{x^4} + \frac{14 t^5}{x^5}
+ \frac{42t^6}{x^6} +  \frac{132t^7}{x^7} + \frac{429t^8}{x^8}+\cdots \biggr),
\een
where the sequence $1,2,5,24,42, 132, 429$ are the Catalan numbers.
In terms of he correlators,
we have shown that:
\be
\corr{p_{2n}}_0^c = \frac{1}{2} \frac{(2n)!}{n!(n+1)!} t^{n+1}.
\ee

\subsection{Computation of genus zero two-point function by Virasoro constraints}

By   equation \eqref{eqn:L0},
\ben
\corr{p_2 \cdot p_{2n}  }^c_{0}  = 2n \cdot \corr{ p_{2n} }_{g}^c ,
\een
and so
\ben
&& \sum_{n=1}^\infty \corr{p_2 \cdot  p_{2n} }^c_{0} \frac{1}{x_2^{-n-1}}
=  \sum_{n \geq 2} 2n \cdot \corr{ p_{2n}   }_{0}^c  \frac{1}{x_2^{n+1}} \\
& = & - \frac{d}{dx_2} \biggl(2x_2 \sum_{n \geq 2} \corr{ p_{2n}}_{0}^c  \frac{1}{x_2^{n+1}} \biggr) \\
& = & -\frac{d}{dx_2} \biggl( \frac{x_2}{2}
 (1-\frac{2t}{x_2}- \sqrt{1-\frac{4t}{x_2}}) \biggr) \\
& = & \frac{1-\frac{2t}{x_2}}{2\sqrt{1-\frac{4t}{x_2}}} -\frac{1}{2}.
\een
It follows that
\ben
\sum_{n=1}^\infty \corr{p_2 \cdot  p_{2n} }^c_{0} \frac{1}{x_2^{-n-1}}
&=&\frac{t^2}{x_2^2}+\frac{4t^3}{x_2^3}
+\frac{15t^4}{x_2^4} + \frac{56t^5}{x_2^5}
+\frac{210t^6}{x_2^6} +\frac{792t^7}{x_2^7}
+\cdots,
\een
where the sequences of the numbers $1$, $4$, $15$, $56$,
$210$, $792$ are the sequence A001791  on \cite{OEIS},
they are given by
\be
a_n = \binom{2n}{n-1} = \frac{(2n)!}{(n-1)!n!}.
\ee
The second equation in the sequence can be written  as
\be
\corr{p_4 \cdot p_{2n} }^c_{0}
 =2 n \cdot \corr{ p_{2n+2} }_{0}^c
 + 2 t \cdot \corr{p_2 \cdot p_{2n}}^c_{0} .
\ee
By taking generating series,
one gets:
\ben
&& \sum_{n \geq 1} \corr{p_4 \cdot p_{2n} }^c_{0} \frac{1}{x_2^{n+1}} \\
& = & \sum_{n \geq 1} \frac{1}{x_2^{n+1}}
\biggl( 2n \cdot \corr{ p_{2n+2} }_{0}^c
 + 2t \cdot \corr{p_2 \cdot p_{2n}}^c_{0} \biggr) \\
& = & 2\sum_{n \geq 2} \frac{n-1}{x_2^n} \corr{ p_{2n} }_{0}^c
+ \sum_{n \geq 1} \frac{2t}{x_2^{n+1}} \cdot \corr{p_2 \cdot p_{2n}}^c_{0}  \\
& = & - 2\frac{d}{dx_2} \biggl(x_2^2 \sum_{n \geq 2} \frac{1}{x_2^{2n+2}} \corr{ p_{2n} }_{0}^c  \biggr)
+ 2t \sum_{n \geq 1} \corr{p_2 \cdot p_{2n}}^c_{0}  \frac{1}{x_2^{n+1}} \\
& = & - 2\frac{d}{dx_2} \biggl(x_2^2 \biggl(\frac{1}{4}
(1-\frac{2t}{x_2}- \sqrt{1-\frac{4t}{x_2}}) - \frac{t^2}{2x_2^2} \biggr) \biggr)
+2 t \biggl(\frac{1-\frac{2t}{x_2}}{2\sqrt{1-\frac{4t}{x_2}}} -\frac{1}{2} \biggr)  \\
& = & x_2 \biggl( \frac{1-\frac{2t}{x_2} - \frac{2t^2}{x_2^2}}{\sqrt{1-\frac{4t}{x_2}}} - 1 \biggr).
\een
Therefore,
\ben
&& \sum_{n \geq 1} \corr{p_4 \cdot p_{2n} }^c_{0} \frac{1}{z_2^{n+1}}  \\
& = & 2 \biggl( \frac{2t^3}{x^2}+\frac{9t^4}{x^3}
+\frac{36t^5}{x^4}+\frac{140t^6}{2^5x^5}
+\frac{540t^7}{x^6}+\frac{2079t^8}{x^7}
+\frac{8008t^9}{x^8}+\frac{30888t^{10}}{x^9}
+ \cdots \biggr),
\een
where the numbers $2$, $9$, $36$,
$140$, $540$, $2079$, $8008$, $30888$
are the sequence A007946 on \cite{OEIS},
and they are given by
\be
a_n = \frac{3(2n)!}{(n-1)!n!(n+2)}.
\ee

\begin{prop}
The following formula holds:
\be \label{eqn:G02}
 G_{0,2}(x_1,x_2)
=  \frac{1-\frac{2t}{x_1}-\frac{2t}{x_2}}
{2(x_1-x_2)^2\sqrt{(1-\frac{4t}{x_1})(1-\frac{4t}{x_2})}}
- \frac{1}{2(x_1-x_2)^2}.
\ee
\end{prop}

\begin{proof}
For $m \geq 4$,
\be
\corr{p_ {2m} \cdot  p_{2n} }^c_{0}
= 2n \cdot \corr{ p_{2n+2m-2} }_{0}^c
+ 4 \sum_{k=0}^{m-1}
\corr{p_{2k} \cdot p_{2n} }^c_{0}
\cdot \corr{p_{2m-2-2k} }^c_{0},
\ee
where we use the following convention:
\be
\corr{p_0}_0^c  = \frac{t}{2},
\ee
Therefore,
\ben
G_{0,2}(x_1, x_2)
& = &  \sum_{m,n \geq 1}\corr{p_ {2m} \cdot  p_{2n} }^c_{0} \frac{1}{x_1^{m+1}} \frac{1}{x_2^{n+1}} \\
& = & \sum_{m+n \geq 2} 2n \cdot \corr{ p_{2n+2m-2} }_{0}^c \frac{1}{x_1^{m+1}} \frac{1}{x_2^{n+1}}  \\
& + & 4 \sum_{m \geq 2,n \geq 1}   \frac{1}{x_1^{m+1}} \frac{1}{x_2^{n+1}}
\sum_{k=0}^{m-1}
\corr{p_{2k} \cdot p_{2n} }^c_{0}
\cdot \corr{p_{2m-2-2k} }^c_{0} \\
& = &
2 \sum_{m+n \geq 2} n \cdot \corr{ p_{2n+2m-2} }_{0}^c  \frac{1}{x_1^{m+1}} \frac{1}{x_2^{n+1}}  \\
& + &   4   \sum_{n,k \geq 1}
\frac{1}{x_1^{k+1}} \frac{1}{x_2^{n+1}}
\corr{p_{2k} \cdot p_{2n} }^c_{0}
\cdot \sum_{j \geq 0} \corr{p_{2j} }^c_{0}
\frac{1}{x_1^{j+1}} \\
& = & 2 \sum_{m+n \geq 2} n \cdot \corr{ p_{2n+2m-2} }_{0}^c \frac{1}{x_1^{m+1}} \frac{1}{x_2^{n+1}}  \\
& + & 4 W_{0,1}(z_1) \cdot W_{0,2}(z_1,z_2),
\een
where $W_{0,1}(z_1)$ is defined by:
\be
W_{0,1}(x_1)
= \frac{t}{2x_1} + G_{0,1}(z_1)
= \sum_{j \geq 0} \corr{p_{2j} }^c_{0}
\frac{1}{x_1^{j+1}}
= \frac{1}{4} \biggl(1 - \sqrt{1-\frac{4t}{x_1}}\biggr).
\ee
So we have:
\ben
G_{0,2}(z_1, z_2)
& = & \frac{2}{1-4W_{0,1}(x_1)}
 \cdot \sum_{l \geq 1}
\corr{ p_{2l} }_{0}^c
\sum_{n=1}^{l} \frac{n}{x_1^{l+2-n}}
\frac{1}{x_2^{n+1}}  \\
& = & \frac{2}{\sqrt{1-\frac{4t}{x_1}}}
  \sum_{l \geq 1}
\corr{ p_{2l} }_{0}^c
\sum_{n=1}^{l} \frac{n}{x_1^{l+2-n}}
\frac{1}{x_2^{n+1}} .
\een

Furthermore,
\ben
&&  \sum_{l \geq 1}
\corr{ p_{2l} }_{0}^c
\sum_{n=1}^{l} \frac{n}{x_1^{l+2-n}}
\frac{1}{x_2^{n+1}} \\
& = &  \sum_{l \geq 1}
\corr{ p_{2l} }_{0}^c  \frac{1}{x_1^{l+1}x_2^2}
\sum_{n=1}^{l} \frac{n}{x_1^{1-n}}
\frac{1}{x_2^{n-1}} \\
& = &   \sum_{l \geq 1}
\corr{ p_{2l} }_{0}^c  \frac{1}{x_1^{l+1}x_2^2}
\biggl(
\frac{1-(\frac{x_1}{x_2})^{l}}{(1-\frac{x_1}{x_2})^2}
-\frac{l (\frac{x_1}{x_2})^{l}}{1- \frac{x_1}{x_2}}
\biggr)  \\
& = & \frac{1}{x_1(x_1-x_2)^2} \sum_{l \geq 1}
\corr{ p_{2l} }_{0}^c
\biggl(\frac{1}{x_1^{l}}-\frac{1}{x_2^{l}} \biggr)
+ \frac{1}{x_1(x_1-x_2)}
\sum_{l=1}^\infty \corr{p_{2l}}_0^c  \cdot \frac{l}{x_2^{l+1}}  \\
& = &  \frac{1}{x_1(x_1-x_2)^2}
(x_1G_{0,1}(x_1) - x_2 G_{0,1}(x_2))
- \frac{1}{x_1(x_1-x_2)} \frac{d}{dx_2}(x_2 G_{0,1}(x_2)) \\
& = & \frac{1-\frac{2t}{x_1}-\frac{2t}{x_2}
- \sqrt{(1-\frac{4t}{x_1})(1-\frac{4t}{x_2})}}
{4(x_1-x_2)^2\sqrt{(1-\frac{4t}{x_2})}},
\een
where in the last equality we have used \eqref{eqn:G01}.
Therefore,
we have  proved \eqref{eqn:G02}.
\end{proof}

\begin{corollary}
The two-point function in genus zero has the following expansion:
\be \label{eqn:G02-2}
 G_{0,2}(x_1,x_2)
= 2 \sum_{l=2}^\infty  \frac{t^l}{l }
\sum_{m+n=l-2}  \frac{(2m+1)!}{(m!)^2}
\cdot   \frac{(2n+1)!}{(n!)^2}
\frac{1}{x_1^{m+2}x_2^{n+2}},
\ee
In terms of correlators we have
\be
\corr{p_{2m}p_{2n} }_0^c = \frac{1}{2} \cdot \frac{(2m)!}{(m-1)!m!} \cdot \frac{(2n)!}{(n-1)!n!} \cdot \frac{t^{m+n}}{m+n}.
\ee
\end{corollary}

\begin{proof}
By \eqref{eqn:G02} we have:
\ben
\frac{\pd}{\pd t} G_{0,2}(x_1, x_2)
& = & \frac{\pd}{\pd t} \biggl(  \frac{1-\frac{2t}{x_1}-\frac{2t}{x_2}}
{2(x_1-x_2)^2\sqrt{(1-\frac{4t}{x_1})(1-\frac{4t}{x_2})}}
- \frac{1}{2(x_1-x_2)^2} \biggr) \\
& = & \frac{2t}{x^2y^2} \frac{1}{(1-\frac{4t}{x_1})^{3/2} (1-\frac{4t}{x_2})^{3/2}} \\
& = & \frac{2t}{x^2y^2} \cdot \sum_{m=0}^\infty \frac{(2m+1)!}{(m!)^2} \frac{t^m}{x_1^m}
\cdot \sum_{n=0}^\infty \frac{(2n+1)!}{(n!)^2} \frac{t^n}{x_2^n} \\
& = & \frac{2}{x^2y^2} \cdot \sum_{l=2}^\infty t^{l-1} \sum_{m+n=l-2}  \frac{(2m+1)!}{(m!)^2}
\cdot   \frac{(2n+1)!}{(n!)^2} \frac{1}{x_1^mx_2^n}.
\een
Then \eqref{eqn:G02-2} is proved by integration.

\end{proof}

\subsection{Computation of $n$-point functions in arbitrary genera by Virasoro constraints}

By \eqref{eqn:P2m} we have:
\ben
&& G_{g,n}(x_0, x_1, \dots, x_n) \\
& = & \sum_{m, a_1, \dots, a_n \geq 1}
\corr{p_{2m} \cdot p_{2a_1} \cdots p_{2a_n}}^c_{g}  \frac{1}{x_0^{m+1}}\prod_{i=1}^n \frac{1}{x_i^{a_i+1}} \\
& = & 2\sum_{m, a_1, \dots, a_n \geq 1}
\sum_{j=1}^n a_j \cdot \corr{p_{2a_1} \cdots
p_{2a_j+2m-2} \cdots p_{2a_n}}_{g}^c
\frac{1}{x_0^{m+1}}\prod_{i=1}^n \frac{1}{x_i^{a_i+1}} \\
& + & \sum_{m \geq 2, a_1, \dots, a_n \geq 1} 2t
\corr{p_{2m-2} \cdot p_{2a_1} \cdots p_{2a_n} }^c_{0}
\frac{1}{x_0^{m+1}}\prod_{i=1}^n \frac{1}{x_i^{a_i+1}}\\
& + & 2 \sum_{k=1}^{m-2} \corr{p_{2k}p_{2m-2-2k} \cdot
p_{2a_1} \cdots p_{2a_n}}^c_{g-1}
\frac{1}{x_0^{m+1}}\prod_{i=1}^n \frac{1}{x_i^{a_i+1}} \\
& + & 2\sum_{k=1}^{m-2} \sum_{\substack{g_1+g_2=g\\I_1 \coprod I_2 = [n]}}
\corr{p_{2k} \cdot \prod_{i\in I_1} p_{a_i} }^c_{g_1}
\cdot \corr{p_{2m-2-2k} \cdot \prod_{i\in I_2} p_{2a_i} }^c_{g_2}
\frac{1}{x_0^{m+1}}\prod_{i=1}^n \frac{1}{x_i^{a_i+1}} \\
& - & \frac{1}{8} \delta_{g,1} \delta_{n,0} \frac{1}{x_0^2}.
\een
Denote by $I$, $II$, $III$, $IV$ and $V$ the five lines on the right-hand side of the second equality
respectively.

We now rewrite $I$ in the following fashion.
Note
\ben
&& \sum_{m, a_j \geq 1}  a_j \cdot
\corr{p_{2a_1} \cdots p_{2a_j+2m-2} \cdots p_{2a_n}}_{g}^c
\frac{1}{x_0^{m+1}} \frac{1}{x_j^{a_j+1}} \\
& = & \sum_{b_j \geq 1}  \sum_{\substack{m, a_j \geq 1\\ m+a_j = b_j+1}}
a_j \cdot \corr{p_{2a_1} \cdots p_{2b_j} \cdots p_{2a_n}}_{g}^c
\frac{1}{x_0^{m+1}} \frac{1}{x_j^{a_j+1}} \\
\een
This leads us to an operator
\be
\frac{1}{x_j^{l+1}} \mapsto
\sum_{\substack{m, n \geq 1\\ m+n = l+1}}
n \cdot  \frac{1}{x_0^{m+1}} \frac{1}{x_j^{n+1}}
\ee
For $l \geq 0$.
Because we have:
\ben
&& \sum_{\substack{m, n \geq 1\\ m+n = l+1}}
n \cdot  \frac{1}{x_0^{m+1}} \frac{1}{x_j^{n+1}} \\
& = & \sum_{n=1}^{l} \frac{n}{x_0^{l+2-n}}\frac{1}{x_j^{n+1}}
= \frac{1}{x_0^{l+1}x_j^2}
\sum_{n=1}^{l} \frac{n}{x_0^{1-n}}
\frac{1}{x_j^{n-1}} \\
& = &  \frac{1}{x_0^{l+1}x_j^2}
\biggl(
\frac{1-(\frac{x_0}{x_j})^{l}}{(1-\frac{x_0}{x_j})^2}
-\frac{l (\frac{x_0}{x_j})^{l}}{1- \frac{x_0}{x_j}}
\biggr)  \\
& = & \frac{1}{x_0(x_0-x_j)^2}
\biggl(\frac{1}{x_0^{l}}-\frac{1}{x_j^{l}} \biggr)
+ \frac{1}{x_0(x_0-x_j)}  \cdot \frac{l}{x_j^{l+1}},
\een
this operator can be realized by:
\be
D_{x_0,x_j}f(x_j)
=  \frac{x_0f(x_0) -x_j f(x_j)}{x_0(x_0-x_j)^2}
- \frac{1}{x_0(x_0-x_j)} \frac{d}{d x_j}
\big(x_jf(x_j) \big).
\ee

Now we examine
\ben
III & = & \sum_{k=1}^{m-2} \corr{p_{2k}p_{2m-2-2k} \cdot p_{a_1} \cdots p_{a_n}}^c_{g-1}(t)
\frac{1}{x_0^{m+1}}\prod_{i=1}^n \frac{1}{x_i^{a_i+1}}.
\een
This leads us to the operator:
\be
\frac{1}{u^{k+1}} \frac{1}{v^{l+1}} \mapsto
\frac{1}{x_0^{k+l+2}}
\ee
This operator can be realized by taking the limit:
\be
E_{x_0, u,v} f(u,v) =   \lim_{u \to v} f(u,v)|_{v = x_0}.
\ee

We combine $II$ with the terms in $IV$ with $g_1=0 $ and $|I_1|=0$,
or $g_2=0$ and $|I_2|=0$.
These together give us $4 W_{0,1}(x_0) G_{0,n+1}(x_0,x_1, \dots, x_n)$.
The rest of the terms in $IV$ give us
$$\sum_{\substack{g_1+g_2=g\\I_1 \coprod I_2=[n]}}\;' E_{x_0, u,v}
\biggl( G_{g_1}(u, x_{I_1}) \cdot G_{g_2}(v, x_{I_2}) \biggr).$$

To summarize,
we obtain the following identity:
\ben
&&  G_g(x_0, x_1, \dots, x_n) \\
& = & 2\sum_{j=1}^n D_{x_0,x_j} G_g(x_1, \dots, x_n)
 + 2E_{x_0, u,v} G_{g-1}(u,v, x_1, \dots, x_n) \\
& + & 2\sum_{\substack{g_1+g_2=g\\I_1 \coprod I_2=[n]}}\;' E_{x_0, u,v}
\biggl( G_{g_1}(u, x_{I_1}) \cdot G_{g_2}(v, x_{I_2}) \biggr) \\
& + & 4 W_{0,1}(x_0) G_{0,n+1}(x_0,x_1, \dots, x_n)
- \frac{1}{8} \delta_{g,1} \delta_{n,0} \frac{1}{x_0^2}.
\een
From this we derive the following:

\begin{prop}
Define the renormalized operators $\tilde{D}$ and $\tilde{E}$ as follows:
\begin{align}
\tilde{D}_{x_0,x_j} & = \frac{2}{1-4 W_{0,1}(x_0)} D_{x_0,x_j}, &
\tilde{E}_{x_0, u,v} & = \frac{2}{1-4 W_{0,1}(x_0)} E_{x_0, u,v}.
\end{align}
Then one has:
\be \label{eqn:G-g-n}
\begin{split}
&  G_{g,n+1}(x_0, x_1, \dots, x_n) \\
= & \sum_{j=1}^n \tilde{D}_{x_0,x_j}
G_{g,n}(x_1, \dots, x_n)
 + \tilde{E}_{x_0, u,v} G_{g-1,n+2}(u,v, x_1, \dots, x_n) \\
+ & \sum_{\substack{g_1+g_2=g\\I_1 \coprod I_2=[n]}}\;'
\tilde{E}_{x_0, u,v}
\biggl( G_{g_1,|I_1|+1}(u, x_{I_1}) \cdot
G_{g_2,|I_2|+1}(v, x_{I_2}) \biggr) \\
- & \frac{\delta_{g,1} \delta_{n,0}}{8(1-4 W_{0,1}(x_0))}  \frac{1}{x_0^2}.
\end{split}
\ee
\end{prop}

\subsection{Examples}

We  now present some sample computations  of $G_{g,n}$ using \eqref{eqn:G-g-n}.

\subsubsection{Three-point function in genus zero}

\ben
&&  G_0(x_0, x_1, x_2) \\
& = & \sum_{j=1}^2 \tilde{D}_{x_0,x_j} G_0(x_1, x_2)
+ 2 \tilde{E}_{x_0, u,v}
\biggl( G_0(u, x_1) \cdot G_0(v, x_2) \biggr) \\
& = & \sum_{j=1}^2 \tilde{D}_{x_0,x_j}
\biggl( \frac{1-\frac{2t}{x_1}-\frac{2t}{x_2}}
{2(x_1-x_2)^2\sqrt{(1-\frac{4t}{x_1})(1-\frac{4t}{x_2})}}
- \frac{1}{2(x_1-x_2)^2}   \biggr) \\
& + & 2\tilde{E}_{x_0, u,v} \biggl(
\biggl(\frac{1-\frac{2t}{u}-\frac{2t}{x_1}}
{2(u-x_1)^2\sqrt{(1-\frac{4t}{u})(1-\frac{4t}{x_1})}}
- \frac{1}{2(u-x_1)^2}\bigg) \\
&& \cdot\biggl( \frac{1-\frac{2t}{v}-\frac{2t}{x_2}}
{2(v-x_2)^2\sqrt{(1-\frac{4t}{v})(1-\frac{4t}{x_2})}}
- \frac{1}{2(v-x_2)^2}\biggr) \biggr).
\een
After a complicated computation with the help of Maple,
the following simple formula is obtained:
\be \label{eqn:G-03}
G_0(x_0,x_1,x_2)
= \frac{4t^2}
{x_0^2x_1^2x_2^2((1-4t/x_0)(1-4t/x_1)(1-4t/x_2))^{3/2}}.
\ee
After expanding this in Taylor series in $t$:
\ben
G_0(x_0,x_1,x_2) & = & \frac{4t^2}{x_0^2x_1^2x_2^2}
\prod_{j=0}^2 \sum_{m_j=0}^\infty \frac{(2m_j+3)!}{m_j!} \frac{(2t)^{m_j}}{x_j^{m_j}},
\een
we get:
\be
\corr{p_{2n_1} p_{2n_2} p_{2n_3}}_0^c
= (2t)^{\sum_{j=1}^3 n_j-1} \prod_{j=1}^3 \frac{(2n_j+1)!!}{n_j!}.
\ee

\subsubsection{Four-point function in genus zero}
In this case \eqref{eqn:G-g-n} takes the following form:
\ben
&&  G_0(x_0, x_1, x_2, x_3)
= \sum_{j=1}^3 \tilde{D}_{x_0,x_j} G_0(x_1, x_2, x_3)  \\
& + &  2 \tilde{E}_{x_0, u,v}
\biggl( G_{0}(u, x_1) \cdot G_{0}(v, x_2, x_3) \\
& + &  G_{0}(u, x_2) \cdot G_{0}(v, x_1, x_3)
+ G_{0}(u, x_3) \cdot G_{0}(v, x_1, x_2) \biggr).
\een
A calculation shows that:
\be \label{eqn:G04}
G_{0,4}(x_0, \dots, x_3)
= 24t^2 \frac{e_4-2e_3t+32e_1t^3-256t^4}
{\prod\limits_{j=0}^3 x_j^3 (1- \frac{4t}{x_j})^{5/2}},
\ee
where $e_j$ denotes the $j$-th elementary symmetric polynomial in $x_0, \dots, x_3$.

\ben
G_{0,4}(x_0, \dots, x_3)
& = & \frac{24t^2}{\prod\limits_{j=0}^3 x_j^3} (e_4-2e_3t+32e_1t^3-256t^4)
\prod\limits_{j=0}^3 \sum_{m_j=0}^\infty
\frac{(2m_j+3)!!}{3\cdot m_j!} \frac{(2t)^{m_j}}{x_j^{m_j}} \\
& = & 24t^2
\prod\limits_{j=0}^3 \sum_{m_j=0}^\infty
\frac{(2m_j+3)!!}{3\cdot m_j!} \frac{(2t)^{m_j}}{x_j^{m_j+2}} \\
& - &  48t^3 \sum_{k=0}^3
\prod\limits_{j=0}^3 \sum_{m_j=0}^\infty
\frac{(2m_j+3)!!}{3\cdot m_j!} \frac{(2t)^{m_j}}{x_j^{m_j+2+\delta_{j,k}}} \\
& + & 768t^5 \sum_{k=0}^3
\prod\limits_{j=0}^3 \sum_{m_j=0}^\infty
\frac{(2m_j+3)!!}{3\cdot m_j!} \frac{(2t)^{m_j}}{x_j^{m_j+3-\delta_{j,k}}} \\
& - & 6144t^6
\prod\limits_{j=0}^3 \sum_{m_j=0}^\infty
\frac{(2m_j+3)!!}{3\cdot m_j!} \frac{(2t)^{m_j}}{x_j^{m_j+3}} \\
& = & 24t^2
\prod\limits_{j=0}^3 \sum_{m_j=0}^\infty
\frac{(2m_j +3)!!}{3\cdot m_j!} \frac{(2t)^{m_j}}{x_j^{m_j+2}} \\
& - &  24t^2 \sum_{k=0}^3
\prod\limits_{j=0}^3 \sum_{m_j=0}^\infty
\frac{(2m_j-2\delta_{j,k}+3)!!}{3\cdot (m_j-\delta_{j,k})!} \frac{(2t)^{m_j}}{x_j^{m_j+2}} \\
& + & 96t^2 \sum_{k=0}^3
\prod\limits_{j=0}^3 \sum_{m_j=0}^\infty
\frac{(2m_j+2\delta_{j,k}+1)!!}{3\cdot (m_j-1+\delta_{j,k})!} \frac{(2t)^{m_j}}{x_j^{m_j+2}} \\
& - & 384t^2
\prod\limits_{j=0}^3 \sum_{m_j=0}^\infty
\frac{(2m_j+1)!!}{3\cdot (m_j-1)!} \frac{(2t)^{m_j}}{x_j^{m_j+2}}.
\een
After simplification we get
\ben
G_{0,4}(x_0, \dots, x_3)
= 8t^2
 \sum_{m_0, \dots, m_3=0}^\infty (\sum_{j=0}^3 m_j+3) \prod\limits_{j=0}^3
\frac{(2m_j+1)!!}{m_j!} \frac{(2t)^{m_j}}{x_j^{m_j}}.
\een
In terms of correlators,
\be
\corr{p_{2n_1} \cdots p_{2n_4}}_0^c
= 2^{\sum_{j=1}^4 n_j -1} t^{\sum_{j=1}^4 n_j -2} (n_1+ \cdots + n_4 -1)
\cdot \prod_{j=1} \frac{(2n_j-1)!!}{(n_j-1)!}.
\ee

\subsubsection{One-point function in genus one}

In this case \eqref{eqn:G-g-n} takes the form:
\ben
G_{1,1}(x_0)
& = & -\frac{1}{8x_0^2(1- 4t/x_0)^{1/2}} + \tilde{E}_{x_0, u,v} G_{0,2}(u,v).
\een
It is easy to see that
\ben
E_{x_0,x,y} W_0^{(2)}(x,y)
& = & \frac{2}{\sqrt{1-\frac{4t}{x_0}}}
\lim_{u \to v} \biggl( \frac{1-\frac{2t}{u}-\frac{2t}{v}}
{2(u-v)^2\sqrt{(1-\frac{4t}{u})(1-\frac{4t}{v})}}
- \frac{1}{2(u-v)^2}\biggr)\biggr|_{v \to x_0} \\
& = & \frac{2t^2}{x_0^4(1- 4t/x_0)^{5/2} }.
\een
Therefore,
\be
G_{1,1}(x_0) = -\frac{1}{8x_0^2(1- 4t/x_0)^{1/2}}
+ \frac{2t^2}{x_0^4(1- 4t/x_0)^{5/2} }.
\ee
By expanding into Taylor series in $t$,
\ben
G_{1,1}(x_0)
& = & - \frac{1}{8x_0^2}
\sum_{m=0}^\infty \frac{(2m-1)!!}{m!} \frac{(2t)^m}{x_0^m}
+ \frac{2t^2}{3x_0^4}
\sum_{n=0}^\infty \frac{(2n+3)!!}{n!} \frac{(2t)^n}{x_0^n} \\
& = & - \frac{1}{8x_0^2}
\sum_{m=0}^\infty \frac{(2m-1)!!}{m!} \frac{(2t)^m}{x_0^m}
+ \frac{1}{6x_0^2}
\sum_{m=0}^\infty m(m-1)\frac{(2m-1)!!}{m!} \frac{(2t)^m}{x_0^m},
\een
we get
\ben
\corr{p_{2n}}_1 & = & - \frac{1}{8} \frac{(2n-3)!!}{(n-1)!}(2t)^{n-1}
+ \frac{1}{6}\frac{(2n-3)!!}{(n-3)!} (2t)^{n-1} \\
& = & \frac{(2n-5)\cdot (2n-1)!!}{24 \cdot (n-1)!} (2t)^{n-1}.
\een
A crucial observation which will play an important role below is that one can rewrite $G_{1,1}(x_0)$ as follows:
\be \label{eqn:G11}
G_{1,1}(x_0) = -\frac{1}{8x_0^2(1- 4t/x_0)^{3/2}}
+ \frac{t}{2x_0^3(1- 4t/x_0)^{5/2} }.
\ee

\subsection{General structure of $G_{g,n}(x_1, \dots, x_n)$}

By induction one easily sees that

\begin{prop} \label{prop:Structure}
When $2g-2+n > 0$,
$G_{g,n}(p_1, \dots, p_n)$ have the following form:
\be
\sum_{\substack{a_1, \dots, a_n \geq 2\\b_1, \dots, b_n \in \bZ}}
A^{g,n}_{a_1, \dots, a_n; b_i \geq a_i -1, 1\leq i \leq n}(t)
x_1^{-a_1} \cdots x_n^{-a_n} y_1^{2b_1+1} \cdots y_n^{2b_n+1},
\ee
where $y_i$ is defined by
\be
y_i = - \frac{1}{4} \sqrt{1- \frac{4t}{x_i}}.
\ee
In particular,
they only have poles of odd orders at $y_j =0$.
\end{prop}

\begin{proof}
We have verified the cases of $G_{0,3}$, $G_{1,1}$ and
$G_{0,4}$.
We now use \eqref{eqn:G-g-n} to inductively finish the proof.
We first consider the terms
$\tilde{E}_{x_0, u,v} G_{g-1,n+2}(u,v, x_1, \dots, x_n)$.
They involve
\ben
&& \tilde{E}_{x_0, u,v} (u^{-a_1}y(u)^{-2b_1-1}
\cdot v^{-a_2}y(v)^{-2b_2-1})
= -\half x_0^{-a_1-a_2} y(x_0)^{-2b_1-2b_2-3},
\een
where we have used the fact that
\ben
&& \frac{2}{1-4W_{0,2}(x_0)} = -\frac{1}{2y_0}.
\een
Because $b_1 \geq a_1-1$, $b_2 \geq a_2 -1$,
we have
\ben
&& b_1 + b_2 +1
\geq (a_1-1) + (a_2-1) + 1 = a_1+a_2 - 1.
\een
Similarly,
when $(g_1, |I_1|+1) \neq (0,2)$ and
$(g_2, |I_2|+1) \neq (0,2)$,
$\tilde{E}_{x_0, u,v}
\biggl( G_{g_1,|I_1|+1}(u, x_{I_1}) \cdot
G_{g_2,|I_2|+1}(v, x_{I_2}) \biggr)$
can be treated in the same way.
The rest of the terms are of the following form:
$$
\sum_{j=1}^n \biggl[\tilde{D}_{x_0,x_j}
G_{g,n}(x_1, \dots, x_n)
+ 2\tilde{E}_{x_0, u,v}
\biggl( G_{0,2}(u, x_j) \cdot
G_{g,n}(v, x_1, \dots, \hat{x}_j, \dots, x_n) \biggr)
\biggr],
$$
and so they involve:
\ben
&& D_{x_0, x_j}(x_j^{-a} y_j^{-2b-1}) + 2 \cE_{x_0, u,v} (G_{0,2}(x_j, u) v^{-a} y(v)^{-2b-1}) \\
& = & \frac{x_0 \cdot x_0^{-a}y_0^{-2b-1}
- x_j \cdot x_j^{-a}y_j^{-2b-1}}{x_0(x_0-x_j)^2}
- \frac{1}{x_0(x_0-x_j)}\cdot
\frac{d}{dx_j} \biggl( x_j \cdot x_j^{-a} y_j^{-2b-1}\biggr) \\
& + & 2\biggl(
\frac{1-\frac{2t}{x_0} - \frac{2t}{x_j}}{2(x_0-x_j)^2 (1- \frac{4t}{x_0})^{1/2}(1-\frac{4t}{x_j})^{1/2}}
- \frac{1}{2(x_0-x_j)^2} \biggr) \cdot x_0^{-a} y_0^{-2b-1} \\
& = & \frac{x_0 \cdot x_0^{-a}y_0^{-2b-1}
- x_j \cdot x_j^{-a}y_j^{-2b-1}}{x_0(x_0-x_j)^2} \\
& - & \frac{1}{x_0(x_0-x_j)}\cdot
\biggl( (-a+1)   x_j^{-a} y_j^{-2b-1}
- (2b+1) x_j^{-a+1} y_j^{-2b-2} \cdot \frac{t}{8x_j^2y_j} \biggr) \\
& + &
\frac{(1-\frac{2t}{x_0} - \frac{2t}{x_j})x_0^{-a} y_0^{-2b-1}}{16(x_0-x_j)^2 y_0y_j}
- \frac{x_0^{-a} y_0^{-2b-1}}{(x_0-x_j)^2}    \\
& = & \frac{1}{x_0(x_0-x_j)^2}
\biggl[  -  x_j^{-a+1}y_j^{-2b-1}  \\
& - &  (x_0-x_j) \cdot
\biggl( (-a+1)   x_j^{-a} y_j^{-2b-1}
- (2b+1) x_j^{-a-1} y_j^{-2b-3} \cdot \frac{t}{8} \biggr) \\
& + &
\frac{1}{16} (1-\frac{2t}{x_0} - \frac{2t}{x_j})x_0^{-a+1} y_0^{-2b-2}y_j^{-1} \biggr].
\een
We rewrite the right-hand side as follows:
\ben
& & \frac{x_j^{-a-1} y_j^{-2b-3}x_0^{-a-1} y_0^{-2b-2}}{x_0(x_0-x_j)^2}
\biggl[  -  x_j^2  y_j^2  x_0^{a+1} y_0^{2b+2} \\
& - &  (x_0-x_j) \cdot
\biggl( (-a+1)   x_j y_j^{2}
- (2b+1) \cdot \frac{t}{8} \biggr) x_0^{a+1} y_0^{2b+2}\\
& + &
\frac{1}{16} (1-\frac{2t}{x_0} - \frac{2t}{x_j})x_0^{2}   x_j^{a+1} y_j^{2b+2} \biggr] \\
& = & \frac{x_j^{-a-1} y_j^{-2b-3}x_0^{-a-1} y_0^{-2b-2}}{2^{8b+12} x_0(x_0-x_j)^2}
\biggl[  -  x_j^2 \cdot  (1-\frac{4t}{x_j}) x_0^{a+1} (1-\frac{4t}{x_0})^{b+1} \\
& - &  (x_0-x_j) \cdot
\biggl( (-a+1)   x_j \cdot  (1-\frac{4t}{x_j})
- (2b+1)    \cdot 2t  \biggr) x_0^{a+1} (1-\frac{4t}{x_0})^{b+1}\\
& + & (1-\frac{2t}{x_0} - \frac{2t}{x_j})x_0^{2}   x_j^{a+1} (1-\frac{4t}{x_j})^{b+1} \biggr] \\
& = & \frac{x_j^{-a-1} y_j^{-2b-3}x_0^{-a-1} y_0^{-2b-2}}{2^{8b+12} x_0(x_0-x_j)^2}
\biggl[  -  x_j \cdot  (x_j-4t) x_0^{a-b} (x_0-4t)^{b+1} \\
& - &  (x_0-x_j) \cdot
\biggl( (-a+1)(x_j-4t)
- (2b+1)    \cdot 2t  \biggr) x_0^{a-b} (x_0-4t)^{b+1}\\
& + & (x_0x_j-2tx_0 - 2tx_j)x_0   x_j^{a-b-1} (x_j-4t)^{b+1} \biggr] \\
& = & \frac{  y_j^{-2b-3}  y_0^{-2b-2}}{2^{8b+12} x_0^{b+2}x_j^{b+2}(x_0-x_j)^2}
\biggl[  -  x_j^{b-a+2} \cdot  (x_j-4t)  (x_0-4t)^{b+1} \\
& - &  (x_0-x_j) \cdot
\biggl( (-a+1)(x_j-4t)
- (2b+1)    \cdot 2t  \biggr) x_j^{b-a+1} (x_0-4t)^{b+1}\\
& + & (x_0x_j-2tx_0 - 2tx_j)x_0^{b-a+1}     (x_j-4t)^{b+1} \biggr] \\
& = & \frac{  y_j^{-2b-3}  y_0^{-2b-2}}{2^{8b+12} x_0^{b+2}x_j^{b+2}}
\sum_{\substack{p+q+r=2b-a+2\\0 \leq p, q \leq b} } \alpha_{p,q,r} x_0^p x_j^q t^r.
\een
By this we complete the proof.
\end{proof}

\section{Emergent Geometry of Modified Hermitian One-Matrix Model}

In last Section we have defined the $n$-point function $G_{g,n}$ and derive an algorithm to compute
them by the operators $\tilde{D}$ and $\tilde{E}$.
This algorithm is based on the Virasoro constraints.
Inspired by Eynard-Orantin \cite{EO},
we will reformulate the $n$-point functions as  multilinear diffrentials on a plane algebraic curve,
satisfying the Eynard-Orantin topological recursions in the next Section.
In this Section we show how spectral curve and its geometry naturally appear
from the point of view of emergent geometry.

\subsection{Emergence of the spectral curve}

Recall the Virasor constraints in genus zero are
the following sequence of differential equations:
\bea
&& \half \frac{\pd F_0}{\pd s_2} =  \sum_{k\geq 1}
k s_{2k} \frac{\pd F_0}{\pd s_{2k}} + \frac{t^2}{4},  \label{eqn:L0-0} \\
&& \half \frac{\pd F_0}{\pd s_{2n+2}} =
\sum^{n-1}_{k=1} \frac{\pd F_0}{\pd s_{2k}}
\frac{\pd F_0}{\pd s_{2n-2k}}
+ t \frac{\pd F_0}{\pd s_{2n}} +  \sum_{k\geq 1}
ks_{2k} \frac{\pd F_0}{\pd s_{2k+2n}}, \quad n \geq 1.
\eea
In earlier work the author has developed
the theory of emergent geometry of spectral curves,
associated with Virasoro constraints,
see \cite{Zhou1,Zhou2, Zhou3}.
The starting point is to consider a suitable generating series of the first derivatives of $F_0$
in all coupling constants.
In this case,
we consider:
\be \label{eqn:Special}
y:=\half \sum_{k=1}^\infty k(s_{2k}-\half \delta_{k,1}) x^{k-1}
+ \frac{t}{2x}  + \sum_{k=1}^\infty \frac{1}{x^{k+1}}
\frac{\pd F_0}{\pd s_{2k}}.
\ee
Then we have:
\ben
y^2 & = & \frac{1}{4}
(\sum_{k=1}^\infty k(s_{2k}-\half \delta_{k,1}) x^{k-1})^2
+ \frac{t^2}{4x^2} +
\biggl( \sum_{k=1}^\infty \frac{1}{x^{k+1}}
\frac{\pd F_0}{\pd s_{2k}} \biggr)^2 \\
& + & \frac{t}{2} \sum_{k=1}^\infty k(s_{2k}-\half \delta_{k,1}) x^{k-2}
+ t\sum_{k=1}^\infty \frac{1}{x^{k+2}}
\frac{\pd F_0}{\pd s_{2k}}
+ \sum_{k, l \geq 1}k (s_{2k}- \half\delta_{k,1})
\frac{\pd F_0}{\pd s_{2l}} x^{k-l-2}.
\een
Therefore,
\ben
(y^2)_- & = &
\biggl( \half t(s_{2}-\frac{1}{2})
+ \sum_{l \geq 1}(l+1) s_{2l+2}
\frac{\pd F_0}{\pd s_{2l}} \biggr) x^{-1} \\
& + & \biggl( \sum_{k \geq 1}k (s_{2k}- \half \delta_{k,1})
\frac{\pd F_0}{\pd s_{2k}} + \frac{t^2}{4} \biggr) x^{-2} \\
& + & \sum_{n=1}^\infty \frac{1}{x^{n+2}}
\biggl(\sum_{k=1}^{n-1}\frac{\pd F_0}{\pd s_{2k}}
\frac{\pd F_0}{\pd s_{2n-2k}}
+t\frac{\pd F_0}{\pd s_{2n}}
+ \sum_{k \geq 1}k (s_{2k}- \half \delta_{k,1})
\frac{\pd F_0}{\pd s_{2n+2k}} \biggr).
\een
So by the Virasoro constraints above:
\be \label{eqn:y2minus}
(y^2)_- = \biggl( \half t(s_{2}-\half)
+ \sum_{l \geq 1}(l+1) s_{2l+2}
\frac{\pd F_0}{\pd s_{2l}} \biggr) x^{-1},
\ee
and so
\ben
y^2 & = & \frac{1}{4}
(\sum_{k=1}^\infty k(s_{2k}-\half \delta_{k,1}) x^{k-1})^2
 \\
& + & \frac{t}{2} \sum_{k=1}^\infty k(s_{2k}-\half \delta_{k,1}) x^{k-2}
+ \sum_{l \geq 1} \sum_{k \geq l+1} k s_{2k}
\frac{\pd F_0}{\pd s_{2l}} x^{k-l-2}.
\een
It follows that when $s_{2k} = 0$,
\be \label{eqn:Spectral}
y^2=\frac{1}{16}-\frac{t}{4x}.
\ee
This defines an algebraic curve on $\bC^2$.
We refer to this curve as the {\em spectral curve} of the modified Hermitian one-matrix model.
We say that \eqref{eqn:Special} defines a {\em special deformation} of the spectral curve.
From \eqref{eqn:Spectral} we get:
\be
x = \frac{4t}{1-16y^2}.
\ee
We call the right-hand side the LG superpotential function of the modified Hermitian one-matrix model.

\subsection{Uniqueness of special deformation}

For a formal power series $a(x) = \sum_{n\in \bZ} a_n z^n$,
let
\be
a_{< -1} = \sum_{n < -1} a_n z^n.
\ee
The following result is very easy to prove:

\begin{thm}
For a series of the form
\be
y =  \frac{1}{2} \frac{\pd S(x; \bs)}{\pd x}
+ \frac{t}{2x}+  \sum_{n \geq 0} w_n x^{-n-2},
\ee
where $S$ is the universal action defined by:
\be \label{eqn:Action2}
S(x;\bs) = - \half x  + \sum_{n \geq 1} s_{2n} x^n,
\ee
and each $w_n \in \bC[[s_2, s_4, \dots]]$,
the equation
\be
(y^2)_{<-1} = 0
\ee
 has a unique solution given by:
\ben
&& y = \frac{1}{2} \sum_{n=0}^\infty
(s_{2n} -\half \delta_{n,1}) x^n
+ \frac{t}{2x} +\sum_{n = 1}^\infty
\frac{1}{x^{n+1}} \frac{\pd F_{0}(t)}{\pd s_{2n}}.
\een
\end{thm}

\subsection{Quantization}

The Virasoro constraints are given by the following differential operators:
\bea
&& L^{even}_0 =  \sum_{k\geq 1}
k (s_{2k}- \half \delta_{k,1}) \frac{\pd}{\pd s_{2k}}
+ \frac{t^2}{4g_s^2} - \frac{1}{16}, \\
&& L^{even}_n = g_s^2
\sum^{n-1}_{k=1} \frac{\pd^2}{\pd s_{2k} \pd s_{2n-2k}}
+ t \frac{\pd}{\pd s_{2n}} +  \sum_{k\geq 1}
k (s_{2k}-\half \delta_{k,1}) \frac{\pd}{\pd s_{2k+2n}}, \quad n \geq 1.
\eea
Consider the generating series of these Virasoro operators:
\ben
&& \sum_{n \geq 0} L^{even}_0 x^{-n-2} \\
& = & \biggl( \sum_{k\geq 1}
k (s_{2k}- \half \delta_{k,1}) \frac{\pd}{\pd s_{2k}}
+ \frac{t^2}{4g_s^2} - \frac{1}{16} \biggr) x^{-2}  \\
& + & \sum_{n\geq 1} \biggl( g_s^2
\sum^{n-1}_{k=1} \frac{\pd^2}{\pd s_{2k} \pd s_{2n-2k}}
+ t \frac{\pd}{\pd s_{2n}} +  \sum_{k\geq 1}
k (s_{2k}-\half \delta_{k,1}) \frac{\pd}{\pd s_{2k+2n}} \biggr) \frac{1}{x^{n+2}} \\
& = & \sum_{k \geq 1} \sum_{n \geq 0} \beta_{-k} \beta_{k+n} \frac{1}{x^{n+k}}
+ \frac{\beta_0^2}{x^2} -\frac{1}{16x^2}
+ \sum_{n \geq 1} \biggl( \sum_{k=1}^{n-1} \beta_k \beta_{n-k} + 2 \beta_0 \beta_n\biggr) \frac{1}{x^{n+2}}
\een
where
\begin{align}
\beta_{-k} & = \half g_s^{-1} k(s_{2k}-\half \delta_{k,1}) \cdot, &
\beta_0 & = \frac{t}{2g_s} \cdot, &
\beta_k & =  g_s\frac{\pd}{\pd s_{2k}},
\end{align}
$k > 0$.
They satisfy the Heisenberg commutation relations:
\be
[\beta_k, \beta_l] = \frac{k}{2} \delta_{k, -l}.
\ee
As usual we take $\{\beta_k\}_{k > 0}$ to be annihilators
and take $\{\beta_{-k}\}_{k > 0}$ to be creators,
and one can then define normally ordered products.
With these notations,
\be
\sum_{n \geq 0} L^{even}_0 x^{-n-2}
= :\hat{y}(x)^2: - \frac{1}{16} x^{-2},
\ee
where
\be
\hat{y}(x):=\sum_{k=1}^\infty \beta_{-k} x^{k-1}
+  \frac{\beta_0}{x}  + \sum_{k=1}^\infty \frac{\beta_k}{x^{k+1}},
\ee
Note we have the following well-known OPE:
\be
\hat{y}(z) \hat{y}(w) = :\hat{y}(z)\hat{y}(w): + \frac{1/2}{(z-w)^2}.
\ee
To account for the extra term $- \frac{1}{16} x^{-2}$,
we use the idea of regularized product in \cite{Zhou1}.
In this case one needs to consider the twisted field
\be
\hat{y}^{twist}(x):= x^{1/2} \cdot \hat{y}(x).
\ee
We have the following OPE for the twisted field:
\be
\hat{y}^{twist}(z) \hat{y}^{twist}(w) = :\hat{y}^{twist}(z)\hat{y}^{twist}(w): + \frac{z^{1/2}w^{1/2}}{2(z-w)^2}.
\ee
In particular,
\ben
&& \hat{y}^{twist}(x+\epsilon) \hat{y}^{twist}(x)  \\
& = & :\hat{y}^{twist}(x+\epsilon)\hat{y}^{twist}(x): + \frac{(x+\epsilon)^{1/2}x^{1/2}}{2\epsilon^2} \\
& = & :\hat{y}^{twist}(x+\epsilon)\hat{y}^{twist}(x):
+ \frac{x}{2\epsilon^2} + \frac{1}{4\epsilon} - \frac{1}{16x} + \cdots,
\een
where we have the following expansion:
\ben
(x+\epsilon)^{1/2}x^{1/2} & = & x(1+\epsilon/x)^{1/2}
= x \biggl(1 + \frac{1}{2} \frac{\epsilon}{x} - \frac{1}{8} \frac{\epsilon^2}{x^2} + \cdots\biggr).
\een
We defined the {\em regularized product} $\hat{y}^{twist}(x)^{\odot 2}
= \hat{y}^{twist}(x) \odot \hat{y}^{twist}(x)$ by:
\be
\hat{y}^{twist}(x)^{\odot 2}:
= \lim_{\epsilon\to 0} \biggl( \hat{y}^{twist}(x+\epsilon) \hat{y}^{twist}(x)
- \biggl(\frac{x}{2\epsilon^2} + \frac{1}{4\epsilon} \biggr) \biggr),
\ee
then we have
\be
\hat{y}^{twist}(x)^{\odot 2} = :\hat{y}^{twist}(x)\hat{y}^{twist}(x): - \frac{1}{16x}
= x \biggl(:\hat{y}(x)\hat{y}(x): - \frac{1}{16x} \biggr),
\ee
and so the Virasoro constraints can be reformulated in the following form:
\be
\biggl( \hat{y}^{twist}(x)^{\odot 2} Z\biggr)_{-} = 0.
\ee

\begin{rmk}
It seems to be more natural to consider instead:
\be
y^{twist}(x):=x^{1/2} y(x) = \half \sum_{k=1}^\infty
k(s_{2k}-\half \delta_{k,1}) x^{k-1/2}
+ \frac{t}{2x}  + \sum_{k=1}^\infty \frac{1}{x^{k+1/2}}
\frac{\pd F_0}{\pd s_{2k}}.
\ee
With this choice
\eqref{eqn:y2minus} becomes:
\be
((y^{twist})^2)_- = \biggl( \half t(s_{2}-\half)
+ \sum_{l \geq 1}(l+1) s_{2l+2}
\frac{\pd F_0}{\pd s_{2l}} \biggr) x^{-1},
\ee
and so
\ben
(y^{twist})^2 & = & \frac{x}{4}
(\sum_{k=1}^\infty k(s_{2k}-\half \delta_{k,1}) x^{k-1})^2 \\
& + & \frac{t}{2} \sum_{k=1}^\infty
k(s_{2k}-\half \delta_{k,1}) x^{k-1}
+ \sum_{l \geq 1} \sum_{k \geq l+1} k s_{2k}
\frac{\pd F_0}{\pd s_{2l}} x^{k-l-1}.
\een
When $s_{2k} = 0$,
\be
(y^{twist})^2=\frac{x}{16}-\frac{t}{4}.
\ee
However,
algebraic considerations force us to choose $y$.
It is well-known that if
\be
[L_m, L_n] = (m-n)L_{m+n} + \frac{(m^3-m)c}{12} \delta_{m,-n},
\ee
then the field
\be
L(z) = \sum_{n \in \bZ} L_n z^{-n-2}
\ee
satisfies the OPE:
\be
L(z)L(w) \sim \frac{L'(w)}{z-w}
+ \frac{2L(w)}{(z-w)^2} + \frac{c/2}{(z-w)^4}.
\ee
In other words,
$\sum_{n \geq 0} L_n^{even} x^{-n-2}
= :\hat{y}(x)^2:-\frac{1}{16x^2}$.
We are then forced by these algebraic considerations
to  make the seemingly unnatural choice of taking $\hat{y}(x)$ in the quantum case
and $y(x)$ in the classical case.
\end{rmk}

\section{Emergence of Topological Recursions}

In this Section we show that the recursion relations
in terms of operators $\tilde{D}$ and $\tilde{E}$
are Eynard-Orantin topological recursions
for the spectral curve discussed in last Section.

\subsection{Genus zero one-point function and the spectral curve}

Let us take $s_{2k} = 0$ in \eqref{eqn:Special} to get:
\be
y=- \frac{1}{4} + \frac{t}{2x} + \sum_{k=1}^\infty \frac{1}{x^{k+1}} \frac{\pd F_0}{\pd s_{2k}}\biggl|_{s_{2n}= 0, n\geq 1}.
\ee
By the definition of correlators and $G_{0,1}(x)$,
\be
y = - \frac{1}{4} + \frac{t}{2x} + G_{0,1}(x).
\ee
By the formula \eqref{eqn:G01} for $G_{0,1}$,
\be
y   = - \frac{1}{4} \sqrt{1-\frac{4t}{x}},
\ee
and so
\be \label{eqn:Spectral-MM}
y^2 = \frac{1}{16}- \frac{t}{4x}.
\ee
This recovers the formula \eqref{eqn:Spectral} for the spectral curve of the modified Hermitian one-matrix
model derived in last Section.

The coordinates of a point $p$ on the spectral curve
  is given by two holomorphic functions $x=x(p)$ and $y=y(p)$.
But this curve  is a rational curve in the plane,
it has a global coordinate given by $y$, and
\be
x = \frac{4t}{1-16y^2}.
\ee
There is a natural hyperelliptic structure on this curve: One can define an involution $
p \mapsto \sigma(p)$ by
\be
\sigma(x, y) = (x, -y).
\ee

\subsection{Correlation functions as functions on the spectral curve}

With the introduction of the spectral curve,
one can regard the genus $g$ $n$-point correlation functions
$G_{g,n}(x_1, \dots, x_n)$ as functions on it.
We understand $x$ and $y$ as meromorphic function on the spectral curve.
For a point $p_j$ on it,
we write
\begin{align}
x_j & = x(p_j), & y_j & = y(p_j).
\end{align}
In \S \ref{sec:Correlation} we have used local coordinates $x_1, \dots, x_n$.
The concrete examples computed there can now be translated into
functions in $y_1, \dots, y_n$:
\ben
&& G_{0,1}(y_1) =  \frac{1}{4} - \frac{t}{2x_1} + y_1 = \frac{1}{8} + y_1 +2 y_1^2
= \frac{1}{8} (4y+1)^2, \\
&& G_{0,2}(y_1, y_2) = \frac{1-\frac{2t}{x_1} - \frac{2t}{x_2}}{2 \cdot 4^2 \cdot (x_1-x_2)^2 y_1y_2}
- \frac{1}{2(x_1-x_2)^2} \\
&& \qquad\qquad\quad = \frac{1}{2^{14}} \frac{(1-16y_1^2)^2(1-16y_2^2)^2}{t^2y_1y_2(y_1+y_2)^2}, \\
&& G_{0,3}(y_1,y_2,y_3)
= - \frac{t^2}{2^{16} \cdot x_1^2x_2^2x_3^2(y_1y_2y_3)^3}\\
&& \qquad\qquad\quad = -\frac{1}{2^{28}t^4} \frac{(1-16y_1^2)^2(1-16y_2^2)^2(1-16y_3^2)^2}{y_1^3y_2^3y_3^3}, \\
&& G_{0,4}(y_1, \dots, y_4)
=  24t^2 \frac{e_4-2e_3t+32e_1t^3-256t^4}
{\prod\limits_{j=0}^3 x_j^3 (4y_j)^5}, \\
&& G_{1,1}(y_1) = \frac{1}{2^5x_1^2y_1}
- \frac{t^2}{2^9x_1^4y_1^5 } \\
&& \qquad\qquad
= \frac{(1-16y_1^2)^2}{2^9y_1} -  \frac{1}{2^{17}} \frac{(1-16y_1^2)^4}{t^2y_1^5}.
\een

Now we translate \eqref{eqn:G-g-n} into the following form:
\be  \label{eqn:G-g-n-y}
\begin{split}
&  G_{g,n+1}(y_0, y_1, \dots, y_n) \\
= & \sum_{j=1}^n \cD_{y_0,y_j} G_{g,n}(y_1, \dots, y_n)
 + \cE_{y_0, y,y'} G_{g-1,n+2}(y,y', y_1, \dots, y_n) \\
+ & \sum_{\substack{g_1+g_2=g\\I_1 \coprod I_2=[n]}}\;' \cE_{y_0, y,y'}
\biggl( G_{g_1, |I_1|+1}(y, y_{I_1}) \cdot G_{g_2, |I_2|+1}(y', y_{I_2}) \biggr)
- \frac{\delta_{g,1} \delta_{n,0}}{32y_0}  \frac{1}{x_0^2},
\end{split}
\ee
where $\cD_{p_0,p_j}$ and $\cE_{p_0, q, r}$ are some operators to be determined below.
Recall
\ben
&& \tilde{D}_{x_0,x_j}f(x_j)
=  \frac{2}{1-4W_{0,1}(x_0)} \biggl( \frac{x_0f(x_0) -x_j f(x_j)}{x_0(x_0-x_j)^2}
- \frac{1}{x_0(x_0-x_j)} \frac{d}{d x_j}
\big(x_jf(x_j) \big) \biggr), \\
&& \tilde{E}_{x_0, u,v} f(u,v) =  \frac{2}{1-4W_{0,1}(x_0)} \lim_{u \to v} f(u,v)|_{v = x_0}.
\een
It is clear that the operator $\tilde{E}$ can be translated into the following operator acting on functions in $y$:
\be
\cE_{y_0, y, y'} g(y, y') = - \frac{1}{2y_0} \lim_{y' \to y} f(y,y')|_{y = y_0}.
\ee
On the other hand,
\ben
&& \frac{2}{1-4W_{0,1}(x_0)} \biggl( \frac{x_0f(x_0) -x_j f(x_j)}{x_0(x_0-x_j)^2}
- \frac{1}{x_0(x_0-x_j)} \frac{d}{d x_j}
\big(x_jf(x_j) \big) \biggr) \\
& = & -\frac{1}{2y_0} \biggl[ \frac{\frac{4t}{1-16y_0^2} g(y_0) - \frac{4t}{1-16y_j^2} g(y_j) }{
{\frac{4t}{1-16y_0^2}\biggl(\frac{4t}{1-16y_0^2}-\frac{4t}{1-16y_j^2}\biggr)^2}} \\
& - & \frac{1}{ {\frac{4t}{1-16y_0^2}\biggl(\frac{4t}{1-16y_0^2}-\frac{4t}{1-16y_j^2}\biggr)}}
\cdot \frac{1}{\frac{dx_j}{dy_j}} \cdot \frac{d}{d y_j} \biggl( \frac{4t}{1-16y_j^2} g(y_j) \biggr) \biggr] \\
& = &- \frac{(1-16y_0^2)^2(1-16y_j^2)^2}{16384t^2y_0y_j }
\cdot \frac{2y_j (g(y_0) - g(y_j))- (y_0^2-y_j^2) g'(y_j)}{(y_0^2-y_j^2)^2},
\een
and so
\be
\cD_{y_0, y_j} g(y_j) = -\frac{(1-16y_0^2)^2(1-16y_j^2)^2}{2^{14}t^2y_0y_j }
\cdot \frac{2y_j (g(y_0) - g(y_j))- (y_0^2-y_j^2) g'(y_j)}{(y_0^2-y_j^2)^2}.
\ee
Note
\ben
&& \cE_{y_0, y, y'} \biggl( x^{-a}y^{-2b-1} \cdot x'^{-a'} y^{-2b'-1} \biggr)
= -\frac{1}{2} x^{-a-a'} y^{-2(b+b')-3}, \\
&& \cD_{y_0, y_j}  \biggl( x^{-a}y_j^{-2b-1} \biggr)
= \frac{(1-16y_0^2)^2(1-16y_j^2)^2}{2^{14}t^2y_0y_j } \\
&& \cdot \frac{2y_j \biggl( \biggl(\frac{1-16y_0^2}{4t} \biggr)^ay_0^{-2b-1}
- \biggl(\frac{1-16y_j^2}{4t} \biggr)^ay_j^{-2b-1} \biggr)- (y_0^2-y_j^2)
\frac{d}{dy_j} \biggl[\biggl(\frac{1-16y_j^2}{4t} \biggr)^ay_j^{-2b-1} \biggr].
}{(y_0^2-y_j^2)^2}
\een

\subsection{Examples}

We  now present some sample computations  of $G_{g,n}$ using \eqref{eqn:G-g-n-y}.

\subsubsection{Three-point function in genus zero}

\ben
&&  G_0(y_0, y_1, y_2) \\
& = & \sum_{j=1}^2 \cD _{y_0,y_j} G_0(y_1, y_2)
+ 2 \cE_{y_0, y,y'}
\biggl( G_0(y, y_1) \cdot G_0(y', y_2) \biggr) \\
& = & \sum_{j=1}^2 \tilde{D}_{y_0,y_j}
\biggl( \frac{1}{2^{14}} \frac{(1-16y_1^2)^2(1-16y_2^2)^2}{t^2y_1y_2(y_1+y_2)^2} \biggr) \\
& + & 2\tilde{E}_{y_0, y,y'}
\biggl(\frac{1}{2^{14}} \frac{(1-16y^2)^2(1-16y_1^2)^2}{t^2yy_1(y+y_1)^2} \\
&& \cdot \frac{1}{2^{14}} \frac{(1-16y'^2)^2(1-16y_2^2)^2}{t^2y'y_2(y'+y_2)^2} \biggr).
\een
After a complicated computation with the help of Maple,
the following simple formula is obtained:
\be
G_0(x_0,x_1,x_2)
= \frac{4t^2}
{x_0^2x_1^2x_2^2((1-4t/x_0)(1-4t/x_1)(1-4t/x_2))^{3/2}}.
\ee

\subsubsection{One-point function in genus one}

In this case \eqref{eqn:G-g-n-y} takes the form:
\ben
G_{1,1}(x_0)
& = &  \frac{1}{32x_0^2y_0} + \cE_{y_0, y,y'} G_{0,2}(y,y') \\
& = &  \frac{1}{32x_0^2y_0} - \frac{1}{2y_0}
\lim_{y' \to y} \biggl(  \frac{1}{2^{14}} \frac{(1-16y^2)^2(1-16y'^2)^2}{t^2yy'(y+y')^2}
\biggr)\biggr|_{y \to y_0} \\
& = & \frac{(1-16y_0^2)^2}{2^9t^2y_0} -  \frac{1}{2^{17}} \frac{(1-16y_0^2)^4}{t^2y_0^5}.
\een

\subsection{Multilinear differential forms}

Instead of the functions $G_{g,n}(p_1, \dots, p_n)$,
one can also consider the multilinear differential forms:
\be \label{def:Omega}
W_{g,n}(p_1, \dots, p_n)
= \hat{G}_{g,n}(y_1, \dots, y_n) dx_1 \dots dx_n,
\ee
where $\hat{G}_{g,n}(y_1, \dots, y_n)=G_{g,n}(y_1, \dots, y_n)$ except for the following two exceptional cases:
\bea
&& \hat{G}_{0,1}(y_1) = -\frac{1}{4}+ \frac{t}{2x_1} + G_{0,1}(y_1), \\
&& \hat{G}_{0,2}(y_1,y_2) = \frac{1}{(x_1-x_2)^2}  + G_{0,2}(y_1,y_2).
\eea
Since $\hat{G}_{0,1}(y_1) = y_1$, so we have:
\be \label{eqn:Omega-0-1}
W_{0,1}(p_1) = y_1dx_1.
\ee
By the following computations
\ben
W_{0,2}(p_1,p_2)
& = & \biggl( \frac{1-\half(1-16y_1^2)-\half(1-16y_2^2)}
{32\biggl(\frac{4t}{1-16y_1^2}-\frac{4t}{1-16y_2^2}\biggr)^2y_1y_2}
+ \frac{1}{2\biggl(\frac{4t}{1-16y_1^2}-\frac{4t}{1-16y_2^2}\biggr)^2} \biggr) \\
&& \cdot d\frac{4t}{1-16y_1^2}d\frac{4t}{1-16y_2^2} \\
& = & \frac{dy_1dy_2}{(y_1-y_2)^2},
\een
we get:
\be \label{eqn:Omega-0-2}
W_{0,2}(p_1,p_2)
= \frac{dy_1dy_2}{(y_1-y_2)^2}.
\ee

\subsection{The computation of the recursion kernel}

We use $W_{0,2}$ as the Bergman kernel.
Then
\ben
\int_{q=\sigma(p_2)}^{p_2} B(p_1, q)
& = & \int_{y=-y_2}^{y_2} \frac{dy_1dy}{(y_1-y)^2} \\
& = & \frac{dy_1}{y_1-y} \biggl|_{y=-y_2}^{y_2} = \frac{2y_2dy_1}{y_1^2-y_2^2}.
\een
It follows that
\be
K(p_0, p) = \frac{dy_0}{2(y_0^2-y^2)dx}= \frac{(1-16y^2)^2}{2^8ty(y_0^2-y^2)} \frac{dy_0}{dy}.
\ee
This has poles at $y=0$ and $y= \pm y_0$.
To understand its behavior at $y = \infty$,
let $y =1/z$.
Then
\be
K(p_0, p) = \frac{z(z^2-16)^2}{2^8t(1-y_0^2z^2)}
\frac{dy_0}{dz}.
\ee

\subsection{Eynard-Orantin topological recursions}

Note
\be
dy = \frac{tdx}{8x^2y}, \quad dx = \frac{2^7ty dy}{(1-16y^2)^2}.
\ee
Let us carry out the first few calculations of Eynard-Orantin recursion
for the spectral curve \eqref{eqn:Spectral-MM} with $\omega_{0,1}=W_{0,1}$ and $\omega_{0,2}=W_{0,2}$
given by \eqref{eqn:Omega-0-1} and \eqref{eqn:Omega-0-2} respectively.
\ben
&& \omega_{0,3}(p_0,p_1,p_2) \\
& = & \Res_{p \to p_+}  K(p_0, p)
[W_{0,2}(p, p_1) W_{0,2}(\sigma(p),p_{2})
+ W_{0,2}(p, p_2) W_{0,2}(\sigma(p), p_1)] \\
& = &  \Res_{y \to 0}  \frac{(1-16y^2)^2dy_0}{2^8ty(y_0^2-y^2)dy} \biggl(
\frac{dy dy_1}{y^2-y_1^2} \cdot \frac{-dy dy_2}{y^2-y_2^2}
+ \frac{dydy_2}{y^2-y_2^2} \cdot \frac{-dydy_1}{y^2-y_1^2} \biggr)  \\
& = &  -\frac{dy_0dy_1dy_2}{2^7ty_0^2y_1^2y_2^2}
= - \frac{t^2dx_0dx_1dx_2}{2^{16}x_0^2x_1^2x_2^2y_0^3y_1^3y_2^3} \\
& = & W_{0,3}(p_0, p_2, p_2).
\een

\ben
\omega_{1,1}(p_0) & = & \Res_{p\to p_+} K(p_0, p)
W_{0,2}(p, \sigma(p))  \\
& = & \Res_{y\to 0} \frac{(1-16y^2)^2dy_0}{2^8ty(y_0^2-y^2)dy} \cdot
\frac{-(dy)^2}{4y^2} \\
& = & - \frac{1- 2^5 y_0^2}{2^{10}t y_0^4} dy_0
=  - \frac{1- 2^5 y_0^2}{2^{10} t y_0^4} \cdot \frac{t}{8x_0^2y_0}dx_0  \\
& = & \biggl( \frac{1}{2^8x_0^2y_0^3} - \frac{1}{2^{13}x_0^2y_0^5} \biggr) dx_0\\
& = & \biggl( \frac{1}{2^9x_0^2y_0^3} - \frac{t}{2^{11}x_0^3y_0^5} \biggr) dx_0\\
& = & W_{1,1}(p_0).
\een
\ben
&& \omega_{0,4}(p_0,p_1,p_2,p_4) \\
& = & \Res_{p \to p_+}  K(p_0, p)
[W_{0,2}(p, p_1) W_{0,3}(\sigma(p),p_{2},p_3) \\
& + & W_{0,2}(p, p_2) W_{0,2}(\sigma(p), p_1, p_3)
+ W_{0,2}(p, p_3) W_{0,2}(\sigma(p), p_1, p_2) \\
& + & W_{0,3}(p, p_1, p_2) W_{0,2}(\sigma(p),p_3)
+ W_{0,3}(p, p_1,p_3) W_{0,2}(\sigma(p), p_2) \\
& + & + W_{0,3}(p, p_1,p_3) W_{0,2}(\sigma(p), p_2) ] \\
& = &  \Res_{p \to p_+}
\frac{dy_0}{(y_0^2-y^2)dx}
\biggl[ \frac{dy dy_1}{(y-y_1)^2}
 \cdot \frac{-t^2dx dx_2dx_3}{1024x^2x_2^2x_3^2(-y)^3y_2^3y_3^3} + perm. \\
& + &  \frac{-t^2dx dx_2dx_3}{2^{10}x^2x_2^2x_3^2y^3y_2^3y_32^3}
\cdot \frac{d(-y) dy_1}{(-y-y_1)^2}
 \cdot+ perm. \biggr] \\
& = &  \Res_{y \to 0}
\frac{dy_0}{(y_0^2-y^2)}
\biggl[ \frac{dy dy_1}{(y-y_1)^2}
 \cdot \frac{t^2 dx_2dx_3}{1024(\frac{2t}{1-4y^2})^2x_2^2x_3^2y^3y_2^3y_3^3} + perm. \\
& + &  \frac{t^2 dx_2dx_3}{1024(\frac{2t}{1-4y^2})^2x_2^2x_3^2y^3y_2^3y_3^3}
\cdot \frac{dy dy_1}{(y+y_1)^2} + perm. \biggr] \\
& = & 2 \cdot \frac{dy_0dy_1dx_2dx_3}{2^{12}x_2^2x_3^2y_2^3y_3^3}\cdot
\frac{-8y_0^2y_1^2+3y_0^2+y_1^2}{y_0^4y_1^4} + perm. \\
& = & \frac{t^2dx_0dx_1dx_2dx_3}{2^{15}x_0^2x_1^2x_2^2x_3^2y_0^5y_1^5y_2^5y_3^5}\cdot
y_2^2y_3^2(-8y_0^2y_1^2+3y_0^2+y_1^2) + perm. \\
& = & 3t^2 \cdot  \frac{e_4-e_3t+4e_1t^3-16t^4}{2^{20}\prod\limits_{j=0}^3 x_j^3 y_j^5}
dx_1\cdots dx_4.
\een

\begin{thm} \label{Thm:Main1}
The multi-linear differential forms $W_{g,n}(p_1, \dots, p_n)$
defined by \eqref{def:Omega} satisfy the Eynard-Orantin topological recursions
given by the spectral curve
\be
y^2 = \frac{1}{16} - \frac{t}{4x}.
\ee
I.e.,
we have
\be
W_{g,n}(p_1, \dots, p_n)
= \omega_{g,n}(p_1, \dots, p_n).
\ee
\end{thm}

\begin{proof}
We have explicitly check the case of $\omega_{0,3}$ and $\omega_{1,1}$.
We now show that other case can be checked by induction.
By \eqref{eqn:EO} and the induction hypothesis,
\ben
&& \omega_{g,n+1}(p_0,p_1, \dots, p_n) =
\Res_{y\to 0} K(y_0, y)  \\
&& \cdot  \biggl[ \hat{G}_{g-1, n+2}(y, -y, y_{[n]}) dxdx   \\
& + & \sum^g_{h=0} \sum_{I \subset [n]}'
\hat{G}_{h,|I|+1}(y, y_I) \hat{G}_{g-h, n-|I|+1}(-y, y_{[n]-I}) dxdx \biggr] dx_1 \cdots dx_n.
\een
By \eqref{eqn:G-g-n-y} and \eqref{def:Omega},
\ben
&& W_{g,n+1}(y_0, \dots, y_n) \\
& = & G_g(y_0, y_1, \dots, y_n) dx_0 \cdots dx_n  \\
& = & \sum_{j=1}^n \cD_{y_0,y_j} G_g(y_1, \dots, y_n) \cdot dx_0 \cdots dx_n \\
& + & \cE_{y_0, y,y'} G_{g-1}(y,y', y_1, \dots, y_n)\cdot  dx_0 \cdots dx_n \\
& + & \sum_{\substack{g_1+g_2=g\\I_1 \coprod I_2=[n]}}\;' \cE_{y_0, y,y'}
\biggl( G_{g_1}(y, y_{I_1}) \cdot G_{g_2}(y', y_{I_2}) \biggr) \cdot  dx_0 \cdots dx_n \\
& - & \frac{\delta_{g,1} \delta_{n,0}}{32y_0}  \frac{1}{x_0^2} dx_0,
\een
By comparing these two recursion,
we see that it suffices to show that
\be \label{eqn:1}
\begin{split}
& \Res_{y\to 0} K(y_0, y)
\cdot  \biggl[ G_{g-1, n+2}(y, -y, y_{[n]}) dxdx  \biggr] \\
= & \cE_{y_0, y,y'} G_{g-1}(y,y', y_1, \dots, y_n) \cdot dx_0,
\end{split}
\ee
and when $(h, |I|+1) \neq (0,2)$ and $(g-h, n-|I|+1) \neq (0,2)$,
\be \label{eqn:2}
\begin{split}
& \Res_{y\to 0} K(y_0, y)  \biggl[
G_{h,|I|+1}(y, y_I) G_{g-h, n-|I|+1}(-y, y_{[n]-I}) dxdx \biggr] \\
= & \cE_{y_0, y,y'} \biggl[
G_{h,|I|+1}(y, y_I) G_{g-h, n-|I|+1}(y', y_{[n]-I}) \biggr] \cdot dx_0,
\end{split}
\ee
and furthermore
\be \label{eqn:3}
\begin{split}
& \Res_{y\to 0} K(y_0, y)  \biggl[
\hat{G}_{0,2}(y, y_j)dxdx_j \cdot G_{g, n}(-y, y_1, \dots, \hat{y}_j, \dots,y_n) dx \biggr] \\
+ & \Res_{y\to 0} K(y_0, y)  \biggl[
\hat{G}_{0,2}(-y, y_j)dxdx_j G_{g, n}(y, y_1, \dots, \hat{y}_j, \dots,y_n) dx \biggr] \\
= & \biggl(\cD_{y_0, y_j} \biggl[ G_{g, n}(y_1, \dots, y_n) \biggr] \\
+ &\cE_{y_0, y,y'} \biggl[G_{0,2}(y,y_j)G_{g,n}(y', y_1, \dots, \hat{y}_j, \dots,y_n) \\
& \quad \quad + G_{g,n}(y, y_1, \dots, \hat{y}_j, \dots,y_n) G_{0,2}(y',y_j)\biggr] \biggr)\cdot dx_0dx_j.
\end{split}
\ee

We now prove \eqref{eqn:1} and \eqref{eqn:2} by Cauchy's residue theorem.
Indeed,
\ben
&& \Res_{y\to 0} K(y_0, y) \cdot  \biggl[ G_{g-1, n+2}(y, -y, y_{[n]}) dxdx  \biggr]\\
&= & \Res_{y\to 0}
\frac{(1-16y^2)^2}{2^8ty(y_0^2-y^2)} \frac{dy_0}{dy}  G_{g-1, n+2}(y, -y, y_{[n]}) dxdx \\
& = & \Res_{y\to 0} \biggl[
\frac{(1-16y^2)^2}{2^8ty(y_0^2-y^2)} \frac{dy_0}{dy} \\
&& \cdot
\sum_{\substack{a, a', a_1, \dots, a_n \geq 2\\b,b', b_1, \dots, b_n  }}
A^{g-1,n+2}_{a,a',a_1, \dots, a_n; b,b',b_1, \dots, b_n}(t) x^{-a} y^{-2b-1} \cdot x^{-a'} (-y)^{-2b'-1} \\
&& \cdot \prod_{i=1}^n x_i^{-a_i} y_i^{-2b_i-1}
\cdot \biggl( \frac{2^7ty dy}{(1-16y^2)^2} \biggr)^2 \biggr] \\
& = & -\sum_{\substack{a, a', a_1, \dots, a_n \geq 2\\b,b', b_1, \dots, b_n  }}
A^{g-1,n+2}_{a,a',a_1, \dots, a_n; b,b',b_1, \dots, b_n}(t)
\cdot \prod_{i=1}^n x_i^{-a_i} y_i^{-2b_i-1}   \\
&& \cdot \frac{dy_0}{2^{2a+2a'-6}t^{a+a'-1}} \Res_{y\to 0} \biggl[
\frac{(1-16y^2)^{a+a'-2}}{y^{2b+2b'+1}(y_0^2-y^2)}   dy \biggr],
\een
where $b \geq a-1$, $b' \geq a'-1$, $b_i \geq a_i-1$, $i=1, \dots, n$ by Proposition \ref{prop:Structure}.
By Cauchy residue theorem,
the residue at $y=0$ can be computed by the residue at $y= \pm y_0$
and the residue at $y = \infty$.
Because
\ben
\frac{(1-16y^2)^{a+a'-2}}{y^{2b+2b'+1}(y_0^2-y^2)}   dy
& = & \frac{(1-16/z^2)^{a+a'-2}}{z^{-(2b+2b'+1)}(y_0^2-\frac{1}{z^2})}   d\frac{1}{z} \\
& = & z^{2b+2b'-2a-2a'+5} \frac{(z^2-16)^{a+a'-2}}{1-y_0^2z^2} dz,
\een
and $b \geq a-1$, $b' \geq a'-1$, the residue at $y = \infty$ vanishes.
We also have
\ben
&& \Res_{y\to y_0} \biggl[
\frac{(1-16y^2)^{a+a'-2}}{y^{2b+2b'+1}(y_0^2-y^2)}   dy \biggr]
+ \Res_{y\to -y_0} \biggl[
\frac{(1-16y^2)^{a+a'-2}}{y^{2b+2b'+1}(y_0^2-y^2)}   dy \biggr] \\
&=  & \frac{(1-16y_0^2)^{a+a'-2}}{y_0^{2b+2b'+2}} .
\een
Therefore,
\ben
&& \Res_{y\to 0} K(y_0, y) G_{g-1, n+2}(y, -y, y_{[n]}) dxdx \\
& = &  -\sum_{\substack{a, a', a_1, \dots, a_n \geq 2\\b,b', b_1, \dots, b_n \geq 0}}
A^{g-1,n+2}_{a,a',a_1, \dots, a_n; b,b',b_1, \dots, b_n}(t)
\cdot \prod_{i=1}^n x_i^{-a_i} y_i^{-2b_i-1}   \\
&& \cdot \frac{dy_0}{2^{2a+2a'-6}t^{a+a'-1}} \frac{(1-16y_0^2)^{a+a'-2}}{y_0^{2b+2b'+2}}  \\
& = & -\sum_{\substack{a, a', a_1, \dots, a_n \geq 2\\b,b', b_1, \dots, b_n \geq 0}}
A^{g-1,n+2}_{a,a',a_1, \dots, a_n; b,b',b_1, \dots, b_n}(t)
\cdot \prod_{i=1}^n x_i^{-a_i} y_i^{-2b_i-1}   \\
&& \cdot   \frac{dx_0}{2x_0^{a+a'}y_0^{2b+2b'+3}}  \\
& = & \cE_{y_0, y,y'} G_{g-1}(y,y', y_1, \dots, y_n) \cdot dx_0.
\een
In the same fashion one  can prove \eqref{eqn:2}.

We now only need to show that
\ben
&& \Res_{y\to 0} K(y_0, y)  \biggl[
\hat{G}_{0,2}(y, y_j) G_{g, n}(-y, y_1, \dots, \hat{y}_j, \dots,y_n) dxdx \biggr]dx_j \\
& + & \Res_{y\to 0} K(y_0, y)  \biggl[
\hat{G}_{0,2}(-y, y_j) G_{g, n}(y, y_1, \dots, \hat{y}_j, \dots,y_n) dxdx \biggr] dx_j\\
& = & \biggl(\cD_{y_0, y_j} \biggl[ G_{g, n}(y_1, \dots, y_n) \biggr]
+ \cE_{y_0, y,y'} \biggl[G_{0,2}(y,y_j)G_{g,n}(y', y_1, \dots, \hat{y}_j, \dots,y_n) \\
&& +  G_{g,n}(y, y_1, \dots, \hat{y}_j, \dots,y_n) G_{0,2}(y',y_j)\biggr] \biggr)\cdot dx_0dx_j.
\een
Indeed,
the left-hand side of  \eqref{eqn:3} is
\ben
&&  \Res_{y\to 0} \biggl[
\frac{(1-16y^2)^2}{2^8ty(y_0^2-y^2)} \frac{dy_0}{dy} \cdot \frac{dydy_j}{(y-y_j)^2} \\
&& \cdot
\sum_{\substack{a_1, \dots, a_n \geq 2\\b_1, \dots, b_n \geq 0}}
A^{g,n}_{a_1, \dots, a_n; b_1, \dots, b_n}(t) x^{-a_j} (-y)^{-2b_j-1}
 \cdot \prod_{i \neq j} x_i^{-a_i} y_i^{-2b_i-1}
\cdot  \frac{2^7ty dy}{(1-16y^2)^2}   \biggr] \\
& + &  \Res_{y\to 0} \biggl[
\frac{(1-16y^2)^2}{2^8ty(y_0^2-y^2)} \frac{dy_0}{dy} \cdot \frac{d(-y)dy_j}{(-y-y_j)^2} \\
&& \cdot
\sum_{\substack{a_1, \dots, a_n \geq 2\\b_1, \dots, b_n \geq 0}}
A^{g,n}_{a_1, \dots, a_n; b_1, \dots, b_n}(t) x^{-a_j} y^{-2b_j-1}
 \cdot \prod_{i \neq j} x_i^{-a_i} y_i^{-2b_i-1}
\cdot  \frac{2^7ty dy}{(1-16y^2)^2}   \biggr] \\
& = & -\half \sum_{\substack{a_1, \dots, a_n \geq 2\\b_1, \dots, b_n \geq 0}}
A^{g,n}_{a_1, \dots, a_n; b_1, \dots, b_n}(t) \prod_{i \neq j} x_i^{-a_i} y_i^{-2b_i-1} \\
&& \cdot \Res_{y\to 0} \biggl[
\frac{x^{-a_j} y^{-2b_j-1}}{  (y_0^2-y^2)}  \cdot \frac{dy }{(y-y_j)^2}
+ \frac{x^{-a_j} y^{-2b_j-1}}{  (y_0^2-y^2)}  \cdot \frac{dy }{(-y-y_j)^2}      \biggr] dy_jdy_0,
\een
so the computation of the left-hand is reduced to the computation of
\ben
A & = & -\half \Res_{y\to 0} \biggl[
\frac{x^{-a_j} y^{-2b_j-1}}{  (y_0^2-y^2)}  \cdot \frac{dy }{(y-y_j)^2}
+ \frac{x^{-a_j} y^{-2b_j-1}}{  (y_0^2-y^2)}  \cdot \frac{dy }{(-y-y_j)^2} \biggr].
\een
In the same way,
the right-hand side is reduced to
\ben
B & = & \biggl( \cD_{y_0, y_j} \biggl[ x_j^{-a_j} y_j^{-2b_j-1}  \biggr]
+ \cE_{y_0, y,y'} \biggl[G_{0,2}(y,y_j) x'^{-a_j} y'^{-2b_j-1} \\
&& +  x^{-a_j} y^{-2b_j-1} G_{0,2}(y',y_j)\biggr] \biggr) \cdot \frac{dx_0}{dy_0} \cdot \frac{dx_j}{dy_j} \\
& = & \biggl( \cD_{y_0, y_j} \biggl[ x_j^{-a_j} y_j^{-2b_j-1}  \biggr] \\
& + & \cE_{y_0, y,y'} \biggl[ \frac{1}{2^{14}} \frac{(1-16y^2)^2(1-16y_j^2)^2}{t^2yy_j(y+y_j)^2}
\cdot  x'^{-a_j} y'^{-2b_j-1} \\
& + &  x^{-a_j} y^{-2b_j-1} \cdot  \frac{1}{2^{14}}
\frac{(1-16y'^2)^2(1-16y_j^2)^2}{t^2y'y_j(y'+y_j)^2} \biggr]  \biggr) \cdot \frac{dx_0}{dy_0}
\cdot \frac{dx_j}{dy_j}  \\
& = & \biggl(- \frac{(1-16y_0^2)^2(1-16y_j^2)^2}{2^{14}t^2y_0y_j } \\
&& \cdot \frac{2y_j \biggl(x_0^{-a_j} y_0^{-2b_j-1}
- x_j^{-a_j}y_j^{-2b_j-1} \biggr)- (y_0^2-y_j^2)
\frac{d}{dy_j} \biggl[x_j^{-a_j}y_j^{-2b_j-1} \biggr]}{(y_0^2-y_j^2)^2} \\
& - &  \frac{2}{2y_0}  \frac{1}{2^{14}} \frac{(1-16y_0^2)^2(1-16y_j^2)^2}{t^2y_0y_j(y_0+y_j)^2}
\cdot  x_0^{-a_j} y_0^{-2b_j-1}  \biggr)\cdot
\frac{2^{14}t^2y_0y_j}{(1-16y_0^2)^2(1-16y_j^2)^2}  \\
& = &  - \frac{2y_j \biggl(x_0^{-a_j} y_0^{-2b_j-1}
- x_j^{-a_j}y_j^{-2b_j-1} \biggr)- (y_0^2-y_j^2)
\frac{d}{dy_j} \biggl[x_j^{-a_j}y_j^{-2b_j-1} \biggr]}{(y_0^2-y_j^2)^2} \\
& - &   \frac{1}{y_0(y_0+y_j)^2}
\cdot  x_0^{-a_j} y_0^{-2b_j-1}.
\een
and so we only need to show that
\be
A = B.
\ee
This can be reduced to the following identity:
\ben
&& \half \Res_{y\to 0} \biggl[
\frac{y^{-2n-1}}{  (y_0^2-y^2)}  \cdot \frac{dy }{(y-y_j)^2}
+ \frac{ y^{-2b-1}}{  (y_0^2-y^2)}  \cdot \frac{dy }{(-y-y_j)^2} \biggr] \\
& = &  \frac{2y_j \biggl(y_0^{-2n-1} - y_j^{-2n-1} \biggr)- (y_0^2-y_j^2)
\frac{d}{dy_j} \biggl[y_j^{-2n-1} \biggr]}{(y_0^2-y_j^2)^2}
 +  \frac{1}{y_0(y_0+y_j)^2} \cdot y_0^{-2n-1},
\een
for $ n \geq -1$.
But this is very easy to prove.
Note
\ben
&& \half \biggl( \frac{1}{  (y_0^2-y^2)}  \cdot \frac{1 }{(y-y_j)^2}
+ \frac{1}{  (y_0^2-y^2)}  \cdot \frac{1 }{(-y-y_j)^2} \biggr)  \\
& = & \frac{1}{y_0^2-y^2} \cdot \frac{y_j^2+y^2 }{(y_j^2-y^2)^2}
= \sum_{n=0}^\infty \frac{y^{2n}}{y_0^{2n+2}y_j^{2n+2}} \sum_{k=0}^n (2n+1-2j)y_0^{2n-2k}y_j^{2k},
\een
therefore, the left-hand side is equal to
\ben
&&  \frac{1}{y_0^{2n+2}y_j^{2n+2}} \sum_{k=0}^n (2n+1-2j)y_0^{2n-2k}y_j^{2k}.
\een
The right-hand side is equal to
\ben
&& \frac{2y_j \biggl(y_0^{-2n-1} - y_j^{-2n-1} \biggr)  + (2n+1) (y_0^2-y_j^2) y_j^{-2n-2} }{(y_0^2-y_j^2)^2}
 +  \frac{1}{y_0(y_0+y_j)^2} \cdot y_0^{-2n-1} \\
& = & \frac{1}{y_0^{2n+2}y_j^{2n+2}(y_0^2-y_j^2)^2}
\cdot \biggl(2y_0y_j^2 ( y_j^{2n+1}-y_0^{2n+1}) + (2n+1) (y_0^2-y_j^2) y_0^{2n+2} \\
& + & (y_0-y_j)^2 y_j^{2n+2} \biggr) \\
& = & \frac{1}{y_0^{2n+2}y_j^{2n+2}(y_0^2-y_j^2)^2}
\cdot \biggl(  (2n+1) y_0^{2n+4}-(2n+3)y_j^2 y_0^{2n+2} \\
& + & y_0^2 y_j^{2n+2}   + y_j^{2n+4}\biggr) \\
& = & \frac{1}{y_0^{2n+2}y_j^{2n+2}} \sum_{k=0}^n (2n+1-2j)y_0^{2n-2k}y_j^{2k}.
\een
This completes the proof.
\end{proof}

 \section{Relationship with Intersection Numbers}

In this Section we relate the $n$-point functions ofthe modified GUE partition function
with even couplings to intersection numbers.

\subsection{Local Airy coordinate near the branchpoint}

Recall the involution of the spectral curve is given by $\sigma: (x, y) \to (x, -y)$.
It has only one fixed point: $(x, y) = (4t, 0)$.
One can introduce the local Airy coordinate $\zeta$ so that
\be
x = 4t + \zeta^2.
\ee
In other words,
\be \label{def:zeta-x-y}
\zeta^2 = x - 4t = \frac{64ty^2}{1-16y^2}.
\ee
We consider the following expansion of the corrected genus zero one-point function:
\ben
y & = & - \biggl(\frac{1}{16}- \frac{t}{4(4t+\zeta^2)} \biggr)^{1/2}
= \pm \biggl(\frac{\zeta^2}{16(4t+\zeta^2)} \biggr)^{1/2}
= \frac{\zeta}{8t^{1/2}(1+\zeta^2/(4t))^{1/2}} \\
& = & \frac{1}{8t^{1/2}} \sum_{n=0}^\infty (-1)^n \frac{(2n-1)!!}{n!} \frac{\zeta^{2n+1}}{2^{3n}t^n}.
\een
This means the variables $t_{2n+3}$ are given by
\be
t_{2n+3} =  (-1)^n \frac{(2n-1)!!}{n!} \frac{1}{2^{3n+3}t^{n+1/2}}.
\ee
For example,
\begin{align*}
t_3 & =  \frac{1}{2^3t^{1/2}}, &
t_5 & = - \frac{1}{2^6t^{3/2}}.
\end{align*}
The conjugate variables $\tilde{t}_k$ are then given by
\be
e^{- \sum_k \tilde{t}_k u^{-k}}
=  \sum_{n \geq 0}  (-1)^n \frac{(2n + 1)!!(2n-1)!!}{n!} \frac{1}{2^{3n+3}t^{n+1/2}} u^{-n}.
\ee

\subsection{Bergman kernel in local Airy coordinate}

Next we consider the expansion of genus zero two-point function
in the local Airy coordinate:
\ben
&& B(p_1,p_2) - \frac{d\zeta_1d\zeta_2}{(\zeta_1-\zeta_2)^2} \\
& = & \biggl( \frac{1-\frac{2t}{x_1}-\frac{2t}{x_2}}{8(x_1-x_2)^2y_1y_2}
+ \frac{1}{2(x_1-x_2)^2} \biggr) dx_1dx_2
- \frac{d\zeta_1d\zeta_2}{(\zeta_1-\zeta_2)^2} \\
& = & \biggl( \frac{1-\frac{2t}{4t + \zeta^2_1}-\frac{2t}{4t + \zeta^2_2}}
{8(\zeta^2_1-\zeta^2_2)^2
\biggl(\frac{\zeta^2_1}{4(4t+\zeta_1^2)} \biggr)^{1/2}
\biggl(\frac{\zeta^2_2}{4(4t+\zeta_2^2)} \biggr)^{1/2}}
+ \frac{1}{2(\zeta^2_1-\zeta^2_2)^2} \biggr) d(4t + \zeta^2_1)d(4t + \zeta^2_2) \\
& - & \frac{d\zeta_1d\zeta_2}{(\zeta_1-\zeta_2)^2} \\
& = &  \biggl(
\frac{(2t\zeta_1^2+2t\zeta_2^2+\zeta_1^2\zeta_2^2)}{2\zeta_1\zeta_2
(\zeta_1^2-\zeta_2^2)^2(4t+\zeta_1^2)^{1/2} (4t+\zeta_2^2)^{1/2}}
+  \frac{1}{2(\zeta^2_1-\zeta^2_2)^2} \biggr) 4\zeta_1\zeta_2 d\zeta_1d\zeta_2 \\
& - & \frac{d\zeta_1d\zeta_2}{(\zeta_1-\zeta_2)^2} \\
& = & \frac{1}{(\zeta^2_1-\zeta^2_2)^2} \biggl(
\frac{2(2t\zeta_1^2+2t\zeta_2^2+\zeta_1^2\zeta_2^2)}{
(4t+\zeta_1^2)^{1/2} (4t+\zeta_2^2)^{1/2}}
- (\zeta_1^2+\zeta_2^2)
\biggr) d\zeta_1d\zeta_2 \\
& = & \frac{1}{(\zeta^2_1-\zeta^2_2)^2} \biggl(
\frac{(\zeta_1^2+\zeta_2^2+\zeta_1^2\zeta_2^2/(2t))}{
(1+\zeta_1^2/(4t))^{1/2} (1+\zeta_2^2/(4t))^{1/2}}
- (\zeta_1^2+\zeta_2^2)
\biggr) d\zeta_1d\zeta_2.
\een
Now we use  the same method used in the proof of \eqref{eqn:G02-2}.
We let $s=1/t$ and $\frac{\pd}{\pd s}$ on the right-hand side of the last equality to get:
\ben
&& - \frac{1}{8(1+\zeta_1^2s/4)^{3/2}(1+\zeta_2^2s/4)^{3/2}} \\
& = & - \frac{1}{8} \sum_{m=0}^\infty \frac{(2m+1)!!}{2^mm!} (-\zeta_1^2s/4)^m \cdot \sum_{n=0}^\infty
\frac{(2n+1)!!}{2^nn!} (-\zeta_2^2s/4)^n,
\een
and so after integration we get:
\be
B(p_1,p_2) = \frac{d\zeta_1d\zeta_2}{(\zeta_1-\zeta_2)^2}
+  \sum_{l=1}^\infty \frac{(-1)^l}{2^{3l}lt^l} \sum_{m+n=l-1} \frac{(2m+1)!!\zeta_1^{2m}}{m!} \cdot
 \frac{(2n+1)!! \zeta_2^{2n}}{n!}.
\ee
In other words,
\be
B_{2k,2l} = \frac{(-1)^{k+l+1}}{2^{3k+3l+3}(k+l+1)t^{k+l+1}} \frac{(2k+1)!!}{k!}\frac{(2l+1)!!}{l!}.
\ee
By \eqref{eqn:hat-Bkl},
\be
\hat{B}_{k,l} = \frac{(-1)^{k+l+1}}{2^{4k+4l+4}(k+l+1)t^{k+l+1}} \frac{(2k+1)!!(2k-1)!!}{k!}
\frac{(2l+1)!!(2l-1)!!}{l!}.
\ee

\subsection{The variables $\zeta_n$}
By \eqref{eqn:zeta-k} we then have
\be
\zeta_k(z) = \frac{(2k - 1)!!}{2^k} \biggl(
\frac{1}{\zeta(z)^{2k+1}}
- \sum_l \frac{(-1)^{k+l+1}(2k+1)!!(2l+1)!! }{2^{3k+3l+3}(k+l+1)k!l!t^{k+l+1}}
 \frac{\zeta(z)^{2l+1}}{2l + 1}  \biggr).
\ee
In particular,
\be
\zeta_0(z) =
\frac{1}{\zeta(z)}
- \sum_l \frac{(-1)^{l+1}(2l+1)!! }{2^{3l+3}(l+1)!t^{l+1}}
 \frac{\zeta(z)^{2l+1}}{2l + 1} .
\ee
The summation over $l$ can be carried out easily to get:
\be
\zeta_0(z) = \frac{(1+\zeta(z)/(4t))^{1/2} }{\zeta(z)} = \frac{1}{8t^{1/2}y(z)}
= - \frac{1}{2t^{1/2} \sqrt{1-\frac{4t}{x}}}.
\ee
We then also have
\be
d \zeta_0(z)   = \frac{t^{1/2}}{x^2(1-\frac{4t}{x})^{3/2}}dx.
\ee

In general,
to take the summation over $l$ in the expression for $\zeta_k$,
we introduce:
\ben
g_k(w, z) &= & \sum_l \frac{(-1)^{k+l+1}(2k+1)!!(2l+1)!! }{2^{3k+3l+3}(k+l+1)k!l!t^{k+l+1}}
 \frac{\zeta(z)^{2l+1}}{2l + 1}  w^{k+l+1}.\emph{}
\een
Then we have
\ben
\frac{\pd}{\pd w} g_k(w,z)
& = & \sum_l \frac{(-1)^{k+l+1}(2k+1)!!(2l+1)!! }{2^{3k+3l+3}k!l!t^{k+l+1}}
 \frac{\zeta(z)^{2l+1}}{2l + 1}  w^{k+l} \\
& = &  \frac{(-1)^{k+1}(2k+1)!! }{2^{3k+3}k!t^{k+1}}\zeta(z) w^k \cdot
\sum_l \frac{(-1)^l(2l+1)!! }{2^{3l}l!t^l}
\frac{\zeta(z)^{2l}}{2l + 1}  w^{l} \\
& = &  \frac{(-1)^{k+1}(2k+1)!! }{2^{3k+3}k!t^{k+1}}\zeta(z)\cdot  w^k \cdot
(1+w\zeta(z)^2/(4t))^{-1/2}.
\een
Now we integrate over $w$ using Lemma \ref{lm:Integration} below
\ben
g_k(w,z)
& = &  \frac{(-1)^{k+1}(2k+1)!! }{2^{3k+3}k!t^{k+1}}\zeta(z) \\
&& \cdot
\biggl(\frac{\sqrt{1+w\zeta(z)^2/(4t)}}{(k+1)\binom{2k+2}{k+1}(\zeta(z)^2/(16t))^{k+1}}
   \sum_{j=0}^k (-1)^{k-j} \binom{2j}{j} (w\zeta(z)^2/(16t))^j \\
& - & \frac{(-1)^k}{(k+1)\binom{2k+2}{k+1}(\zeta(z)^2/(16t))^{k+1}} \biggr) \\
& = &
\frac{(-1)^{k+1}\sqrt{1+w\zeta(z)^2/(4t)}}{\zeta(z)^{2k+1}}
   \sum_{j=0}^k (-1)^{k-j} \binom{2j}{j} (w\zeta(z)^2/(16t))^j
+ \frac{1}{\zeta(z)^{2k+1}} ,
\een
and then set $w=1$ to get:
\be
\begin{split}
\zeta_k(z)& = \frac{(2k - 1)!!}{2^k}
\cdot \frac{\sqrt{1+\zeta(z)^2/(4t)}}{\zeta(z)^{2k+1}}
\sum_{j=0}^k (-1)^{j} \binom{2j}{j}
\frac{\zeta(z)^{2j}}{2^{4j}t^j}.
\end{split}
\ee
Using \eqref{def:zeta-x-y},
we express $\zeta_k$ as a polynomial in $y^{-1}$:
\be
\zeta_k(z) = \frac{(2k - 1)!!}{2^kt^k}
\cdot \frac{1}{8t^{1/2}y}
\sum_{j=0}^k (-1)^{j} \binom{2j}{j}
\frac{1}{2^{4j}}
\biggl(\frac{1}{64y^2} - \frac{1}{4} \biggr)^{k-j},
\ee
and we can also express it as a power series in $x^{-1}$ using:
\be
\begin{split}
\zeta_k(z)& = - \frac{(2k - 1)!!}{2^k}
\cdot  \sum_{j=0}^k (-1)^{j} \binom{2j}{j}
\frac{1}{2^{4j+1}t^{j+1/2}x^{k-j}(1-\frac{4t}{x})^{k-j+1/2}}.
\end{split}
\ee
In particular,
\ben
\zeta_1(z)
& = & -\frac{1}{2}
\cdot  \sum_{j=0}^1 (-1)^{j} \binom{2j}{j}
\frac{1}{2^{4j+1}t^{j+1/2}x^{1-j}(1-\frac{4t}{x})^{3/2-j}} \\
& = & \frac{1}{2^5t^{3/2}(1-\frac{4t}{x})^{1/2}}
- \frac{1}{2^2t^{1/2}x(1-\frac{4t}{x})^{3/2}},
\een
and so
\be
d \zeta_1(z)  =
\biggl(\frac{3}{16t^{1/2}x^2(1-\frac{4t}{x})^{3/2}}
+ \frac{3t^{1/2}}{2x^3(1-\frac{4t}{x})^{5/2}}\biggr)
dx.
\ee

\begin{lm} \label{lm:Integration}
For $n \geq 0$,
\be
\int \frac{x^n}{\sqrt{1+4ax}} dx
=  \frac{\sqrt{1+4ax}}{(n+1)\binom{2n+2}{n+1}a^{n+1}}
\sum_{j=0}^n (-1)^{n-j} \binom{2j}{j} a^jx^j.
\ee
\end{lm}

\begin{proof}
This is very easy to verify: Simply take derivatives in $x$ on both sides.
The right-hand side gives us
\ben
&& \frac{\sqrt{1+4ax}}{(n+1)\binom{2n+2}{n+1}a^{n+1}} \sum_{j=0}^n (-1)^{n-j} \binom{2j}{j} ja^jx^{j-1} \\
& + & \frac{2a}{(n+1)\binom{2n+2}{n+1}a^{n+1}\sqrt{1+4ax}} \sum_{j=0}^n (-1)^{n-j} \binom{2j}{j} a^jx^{j} \\
& = & \frac{1}{(n+1)\binom{2n+2}{n+1}a^{n+1}\sqrt{1+4ax}}
\biggl( 2a \sum_{j=0}^n (-1)^{n-j} \binom{2j}{j} a^jx^{j} \\
&& + (1+4ax) \sum_{j=0}^n (-1)^{n-j} \binom{2j}{j} ja^jx^{j-1} \biggr) \\
& = & \frac{(-1)^n}{(n+1)\binom{2n+2}{n+1}a^{n+1}\sqrt{1+4ax}} \\
&& \cdot \biggl(\sum_{j=0}^n (-1)^j \binom{2j}{j} (2+4j)a^{j+1}x^{j}
- \sum_{j=0}^{n-1} (-1)^j \binom{2j+2}{j+1} (j+1)a^{j+1}x^{j} \biggr) \\
& = & \frac{x^n}{\sqrt{1+4ax}}.
\een
\end{proof}

\begin{remark}
The numbers $n\binom{2n}{n}$ are the sequence A005430 on \cite{OEIS}.
They are called the Ap\'ery numbers.
\end{remark}
 
\subsection{Relationship with intersection numbers}

By combining Theorem \ref{Thm:Main1} with Eynard's result recalled in \S \ref{sec:Eynard},
we get the following:

\begin{thm} \label{Thm:Main2}
For the modified partition function of Hermitian one-matrix model with even couplings,
\be
\begin{split}
& W_{g,n}(z_1, \dots, z_n) \\
= & 2^{3g-3+n} \sum_{d_1+\cdots +d_n \leq 3g-3+n}
\prod_i
d\zeta_{d_i}(z_i) \Corr{e^{\frac{1}{2} \sum_\delta \iota_{\delta*} \hat{B}(\psi, \psi')}
e^{\sum_k \tilde{t}_k \kappa_k} \prod_i \psi_i^{d_i}}_{g,n}
\end{split}
\ee
where $\tilde{t}_k$ are given by:
\be
e^{- \sum_k \tilde{t}_k u^{-k}} 
=  \sum_{n \geq 0}  (-1)^n \frac{(2n + 1)!!(2n-1)!!}{n!} \frac{1}{2^{3n+3}t^{n+1/2}} u^{-n},
\ee
$\hat{B}_{k,l}$ are given by 
\be
\hat{B}_{k,l} = \frac{(-1)^{k+l+1}}{2^{4k+4l+4}(k+l+1)t^{k+l+1}} \frac{(2k+1)!!(2k-1)!!}{k!}
\frac{(2l+1)!!(2l-1)!!}{l!},
\ee
and $\zeta_k$ are given by:
\be
\begin{split}
\zeta_k(z)& = - \frac{(2k - 1)!!}{2^k}
\cdot  \sum_{j=0}^k (-1)^{j} \binom{2j}{j}
\frac{1}{2^{4j+1}t^{j+1/2}x^{k-j}(1-\frac{4t}{x})^{k-j+1/2}}.
\end{split}
\ee
\end{thm}

Let us now check some example.
By \eqref{eqn:omega03} we have
\ben
\omega_{0,3} (z_1, z_2, z_3)
& = & \frac{1}{2t_3}
d\zeta_0(z_1) d\zeta_0(z_2) d\zeta_0(z_3) \\
& = & \frac{1}{2 \cdot \frac{1}{2^3t^{1/2}} }
\frac{t^{1/2}}{x_1^2(1-\frac{4t}{x_1})^{3/2}}dx_1
\cdot \frac{t^{1/2}}{x_2^2(1-\frac{4t}{x_2})^{3/2}}dx_2
\cdot \frac{t^{1/2}}{x_3^2(1-\frac{4t}{x_3})^{3/2}}dx_3 \\
& = & \frac{4t^2}
{x_0^2x_1^2x_2^2((1-4t/x_0)(1-4t/x_1)(1-4t/x_2))^{3/2}}
dx_1dx_2dx_3.
\een
This matches with \eqref{eqn:G-03}.
By \eqref{eqn:omega11}
\ben
\omega_{1,1}(z) & = & \frac{1}{24 t_3} d\zeta_1(z)
+ \biggl(\frac{B_{0,0}}{4t_3}- \frac{t_5}{16t_3^2}
\biggr) d\zeta_0(z) \\
& = & \frac{1}{24 \cdot \frac{1}{2^3t^{1/2}}} \cdot
\biggl( \frac{3}{16t^{1/2}x^2(1-\frac{4t}{x})^{3/2}}
+ \frac{3t^{1/2}}{2x^3(1-\frac{4t}{x})^{5/2}} \biggr)dx\\
& + & \biggl( \frac{-\frac{1}{2^3t} }{4\cdot \frac{1}{2^3t^{1/2}}}
- \frac{- \frac{1}{2^6t^{3/2}}}{16\cdot (\frac{1}{2^3t^{1/2}})^2 } \biggr)
\cdot  \frac{t^{1/2}}{x^2(1-\frac{4t}{x})^{3/2}}dx \\
& = & \biggl( \frac{1}{16x^2(1-\frac{4t}{x})^{3/2}}
+ \frac{t}{2x^3(1-\frac{4t}{x})^{5/2}}
- \frac{3}{16x^2(1-\frac{4t}{x})^{3/2}} \biggr) dx \\
& = & \biggl(- \frac{1}{8x^2(1-\frac{4t}{x})^{3/2}}
+ \frac{t}{2x^3(1-\frac{4t}{x})^{5/2}} \biggr) dx.
\een
This matches with \eqref{eqn:G11}.

\ben
\omega_{0,4}(z_1, z_2, z_3, z_4)
&=& \frac{1}{2t^2_3}
(d\zeta_1(z_1)d\zeta_0(z_2)d\zeta_0(z_3)d\zeta_0(z_4)
+ perm.) \\
& + & \frac{3}{4} \biggl(\frac{B_{0,0}}{t_3^2} - \frac{t_5}{t_3^3} \biggr)
d\zeta_0(z_1)d\zeta_0(z_2)d\zeta_0(z_3)d\zeta_0(z_4) \\
& = &
\frac{1}{2 \cdot (\frac{1}{2^3t^{1/2}})^2 } \cdot
\biggl( \frac{3}{16t^{1/2}x_1^2(1-\frac{4t}{x_1})^{3/2}}
+ \frac{3t^{1/2}}{2x_1^3(1-\frac{4t}{x_1})^{5/2}} \biggr)dx_1\\
&& \cdot  \frac{t^{1/2}}{x_2^2(1-\frac{4t}{x_2})^{3/2}}dx_2
\cdot  \frac{t^{1/2}}{x_3^2(1-\frac{4t}{x_3})^{3/2}}dx_3
\cdot  \frac{t^{1/2}}{x_4^2(1-\frac{4t}{x_4})^{3/2}}dx_4 \\
& + & perm. \\
& + & \frac{3}{4} \biggl( \frac{-\frac{1}{2^3t} }{(\frac{1}{2^3t^{1/2}})^2}
- \frac{- \frac{1}{2^6t^{3/2}}}{ (\frac{1}{2^3t^{1/2}})^3 } \biggr)
\cdot  \prod_{j=1}^4 \frac{t^{1/2}}{x_j^2(1-\frac{4t}{x_j})^{3/2}}dx_j \\
& = & 6t^2 (x_1+4t)(x_2-4t)(x_3-4t)(x_4-4t)
  \prod_{j=1}^4 \frac{dx_j}{x_j^3(1-\frac{4t}{x_j})^{5/2}} \\
& + & perm. \\
& = & 24t^2 \frac{e_4-2e_3t+32e_1t^3-256t^4}
{\prod\limits_{j=0}^3 x_j^3 (1- \frac{4t}{x_j})^{5/2}} dx_1 \cdots dx_4,
\een
where $e_j$ denotes the $j$-th elementary symmetric polynomial in $x_1, \dots, x_4$.
This matches with \eqref{eqn:G04}.

\vspace{.2in}
{\bf Acknowledgements}.
The author is partly supported by NSFC grants 11661131005 and 11890662.

\end{document}